\documentclass[journal,12pt,onecolumn,draftclsnofoot,]{IEEEtran}
\ifCLASSINFOpdf
\else
\fi

\usepackage{graphicx,amssymb,amsthm,amsmath}
\usepackage{epstopdf}
\usepackage{mathtools}

\newtheorem{theorem}{Theorem}[section]

\newtheorem{lemma}[theorem]{Lemma}
\newtheorem{corollary}[theorem]{Corollary}

\newtheorem{definition}[theorem]{Definition}
\newtheorem{remark}[theorem]{Remark}

\begin{document}
%
\title{Asymptotic convergence rate of the \\longest run in an inflating Bernoulli net}
%
%
%

\author{Kai~Ni,
Shanshan~Cao,
and~Xiaoming~Huo
\thanks{School of Mathematics, Georgia Institute of Technology, Atlanta,
GA, 30332 USA. 
  (e-mail: kni0219@gmail.com).}
\thanks{Department of Industrial and System Engineering, Georgia Institute of Technology, Atlanta,
GA, 30332 USA. 
  (e-mail: shancao36@gmail.com).}
\thanks{Department of Industrial and System Engineering, Georgia Institute of Technology, Atlanta,
GA, 30332 USA. 
  (e-mail: huo@gatech.edu).}}
%
%

\markboth{Submitted}%
{Shell \MakeLowercase{\textit{et al.}}: Bare Demo of IEEEtran.cls for IEEE Journals}
%



\maketitle

\begin{abstract}
In image detection, one problem is to test whether the set, though mostly consisting of uniformly scattered points, also contains a small fraction of points sampled from some (a priori unknown) curve, for example, a curve with $C^{\alpha}$-norm bounded by $\beta$.
One approach is to analyze the data by counting membership in multiscale multianisotropic strips, which involves an algorithm that delves into the length of the path connecting many consecutive ``significant'' nodes.
In this paper, we develop the mathematical formalism of this algorithm and analyze the statistical property of the length of the longest significant run.
 The rate of convergence is derived.
Using percolation theory and random graph theory, we present a novel probabilistic model named pseudo-tree model.
Based on the asymptotic results for pseudo-tree model, we further study the length of the longest significant run in an ``inflating'' Bernoulli net.
We find that the probability parameter $p$ of significant node plays an important role: there
is a threshold $p_c$, such that in the cases of $p<p_c$ and $p>p_c$, very different asymptotic behaviors of the length of the significant are observed.
We apply our results to the detection of an underlying curvilinear feature and argue
that we achieve the lowest possible detectable strength in theory.
\end{abstract}


%
\IEEEpeerreviewmaketitle

\section{Introduction}

In application of image detection problems, one class of questions is to determine whether or not some filamentary structures are present in the noisy picture. 
 There is a plethora of available statistical methods that can, in principle, be used for filaments detection and estimation. 
 These include: Principle curves in \cite{Hastie}, \cite{Kegl}, \cite{Sathyakama} and \cite{Smola}; nonparametric, penalized, maximum likelihood in \cite{RTibshirani}; parametric models in \cite{Stoica}; manifold learning techniques in \cite{Roweis}, \cite{JoshuaB} and \cite{HuoandChen}; gradient based methods in \cite{DNovikov} and \cite{Genovese}; methods from computational geometry in \cite{Dey}, \cite{Lee1999} and \cite{Cheng2005}; faint line segment detection in \cite{ExpandVision}; Ship Wakes ``V'' shape detection against a highly cluttered background in \cite {Vdetection} and underlying curvilinear structure in \cite{Aos06Filament}, \cite{AAOMT} and \cite{IEEE}. 
 See also \cite{perim1}, \cite{perim2} and \cite{perim3} for the applications of the percolation theory in this area.

One approach for this type of detection problems works as follows. 
At localized batches, hypothesis testing is run to determine whether this batch may overlap with the underlying structure. 
The hypothesis testing is run while the batch scans through the entire image. 
The intuition is that if there is an embedded structure, then the significant test results must be clustered around the underlying structure. 
The difficulty comes from the fact that there will be many false positives among these tests. We want to take advantage of the fact that the false positive testing results are not clustered, in relative to those that overlap with the underlying feature. 
Our percolation analysis is motivated by the above phenomenon.

Suppose we have an $m$-by-$n$ array of nodes. 
A Bernoulli random variable $X_{i,j}$ is associated with each node $(i,j)$  such that if $X_{i,j}=1$ then the node is significant (or open); otherwise, insignificant (or closed). 
However, we suspect that there is a sequence of nodes, with unknown location or orientation, open or closed with a different probability $p_1>p$.
In \cite{JASA06LongSigRun}, it is shown that the length of the longest significant run, denoted by $\left|L_{0}(m,n)\right|$ throughout the paper, has the following asymptotic rate of  Erd$\ddot{o}$s-R$\acute{e}$nyi type (See \cite{TERLinD})
\begin{equation}\label{lim:constrain}
\lim_{n\to\infty}\frac{\left|L_{0}(m,n)\right|}{\log_{1/\rho(m,p)}n}=1\quad\text{almost surely},
\end{equation}
where $\rho(m,p)$ is a constant depending on $m$ and $p$ and also the structure of the model.

However the limitation of (\ref{lim:constrain}) is that $m$ is always fixed. 
Our paper extends the previous work to derive the convergence rate of the length of the longest significant run in the inflating model i.e., $m\rightarrow\infty$ and $n\rightarrow\infty$ simultaneously. Our theory is related to the percolation theory, in which we will introduce the critical probability $p_c$ and divide our theory into the $p>p_c$ phase and the $p<p_c$ phase. 
For percolation theory, books by Grimmett \cite{GG1989} and Bollob\'{a}s \cite{B:BOP} are good references. Durrett \cite{OSitePerc1984} systematically studies an oriented site percolation model, which is similar to the model in this paper. See also the references therein. 

Applications of the aforementioned can be the following:
\begin{itemize}
\item Detection of filamentary structures in a background of uniform random points in \cite{Aos06Filament}. 
We are given $N$ points that might be uniformly distributed in the unit square $[0,1]^2$. 
We wish to test whether the set, although mostly consisting of the uniformly scattered points, also contains a small fraction $\epsilon_N$ of points sampled from some (unknown a priori) curve with $C^{\alpha}$ norm bounded by $\beta$. 
See also \cite{NetworksPoly} for a more general case.
\item Target tracking problem in \cite{AAOMT}. 
Suppose we have an infrared staring array. A distant moving object will create, upon lengthy exposure, an image of a very faint track against a noisy background. 
We want to detect whether there is such a moving object in an noisy image.
\item Water quality in a network of streams in \cite{waterQuality}. Water quality in a network of streams is assessed by performing a chemical analysis at various locations along the streams. 
As a result, some locations are marked as problematic. 
We may view the set of all tested locations as nodes and connect pairs of adjacent nodes located on the same stream, thereby creating a tree. 
We then assign to each node the value $1$ or $0$, according to whether the location is problematic or not. 
One can then imagine that one would like to detect a path (or a family of paths) upstream of a certain sensitive location, in order to trace the existence of a polluter, or look for the existence of an anomalous path upstream from the root of the system.
\end{itemize}

There is a multitude of applications for which our model is relevant. 
Examples include the detection of hazardous materials \cite{Hill}, target tracking \cite{Lietal} in sensor networks \cite{Culleretal}, and disease outbreak detection \cite{Heffernanetal}. 
Pixels in digital images are also sensors so that many other examples can be found in the literature on image processing such as road tracking \cite{Gemanetal}, fire prevention using satellite imagery \cite {PozoDetal},x and the detection of tumors in medical imaging \cite{McInerney}.

The generalized likelihood ratio test, which is known as the scan statistic in spatial statistics \cite{Kulldorff1997, Kulldorff2001}, is by far the most popular method in practice and is given different names in different fields. 
Most of the methods related to scan statistic assume that the clusters are in some parametric family such as circular \cite{Kulldorff1995}, elliptical \cite{Hobolth2002, Kulldorff2006} or, more generally, deformable templates \cite{Jain1998}, while others do not assume explicit shapes \cite{Duczmal2004, Kulldorff2003, Tango2005}, which leads to nonparametric models.

We consider a nonparametric method based on the percolative properties of the network. 
The most basic approach is based on the size of the largest significant chain of the graph after removing the nodes whose values fall under a given threshold. 
If the graph is a one-dimensional lattice, after thresholding, this corresponds to the test based on the longest run \cite{Balakrishnan2002}, which \cite{JASA06LongSigRun} adapts for path detection in a thin band. 
This test is studied in a series of papers such as \cite{perim1} under the name of maximum cluster test. 
A more sophisticated statistic, which is the upper level set scan statistic, is studied in  \cite{Patil2004, Patil2006, Patil2010}. In its basic form, it scans over the connected components of the graph after thresholding.

Recently, Langovoy \emph{et al} \cite{perim1,perim2,perim3} employs the theory of percolation and random graph to solve the image detection problem. 
However, our methods in this paper are different from the classic percolation theory, since the nodes here are not necessarily independent \emph{a priori}.

Specifically, our work has three advantages.
\begin{enumerate}
\item We can drop the independence assumption among nodes which is the fundamental assumption in the percolation theory.
\item Our work is devoted to researching the asymptotic behavior of the longest left-right significant run in the lattice with a diverging $m$.
\item Our model can be easily adapted to the three or higher dimensional cases with some notations change, though for simplicity, the paper is mostly written based on a $2$-dimensional model.
\end{enumerate}

In practice, our work places a fundamental theory on practical problems involving the length of runs. One direct motivation comes from a statistical detection problem. 
In \cite{Aos06Filament}, the authors proposed a method called the multi-scale significance run algorithm (MSRA) for the detection of curvilinear filaments in noisy images. The main idea is to construct a Bernoulli net. 
Each node has the value of $1$ (\textit{significant}) or $0$ (\textit{insignificant}). 
Two nodes are defined as \textit{connected} if they are neighbors (for example their altitude difference is within $C$), that is, they can simultaneously cover a curve of interest. 
The length of the nodes in the longest significant run is used as a test statistic. If the length of the run exceeds a certain threshold, then we conclude that there exists an embedded curve; otherwise, there is no embedded curve. To formulate this as a well-defined probability problem, we test the null hypothesis of a constant success probability $p$ against the alternative hypothesis that some nodes, being on a filament with unknown location and length, have a greater probability of success $p_1>p$. Under the alternative, the length of the longest significant chain, $\left|L_0(m,n)\right|$, is more likely to exceed (i.e., be greater than) a threshold, which, under the null hypothesis, cannot be exceeded. In the approach of \cite{Aos06Filament} the values of these parameters can be chosen for testing. The question is how to choose these parameters so that the power of the test can be maximized. This becomes a design issue. The relation between $\left|L_0(m,n)\right|$ and other parameters must be understood. 
The choice of parameters in the approach of \cite{Aos06Filament} is sufficient to guarantee a proof of asymptotic optimality; Our research systematically searches the relation between $\left|L_0(m,n)\right|$ and these parameters.

In \cite{JASA06LongSigRun} the authors show that $\rho(m,p)$ in (\ref{lim:constrain}), which is the limit of conditional probability $\rho_n(m,p)$ that there will be an across for $n$ columns conditioning on the fact that there is an across in the previous $\left(n-1\right)$ columns, lies in $(0,1)$ as $n\rightarrow\infty$.  Let  $\mathcal{A}_{c_1,c_2,\delta_1,\delta_2}$ denote the following set
\begin{equation}\label{cond:m}
\mathcal{A}_{c_1,c_2,\delta_1,\delta_2} \vcentcolon= \{(m,n): c_1n^{1+\delta_1}\leq m\leq c_2\exp[n(\phi(p)-\delta_2)]\}.
\end{equation}
The set  $\mathcal{A}_{c_1,c_2,\delta_1,\delta_2}$ essentially states that as the column number $n$ increases, $m$ increases faster than any linear growth of $n$ and slower than some exponential growth of $n$. In our work, we show that in the case of $p<p_c$, as $m\rightarrow\infty$, $n\rightarrow\infty$ and $(m,n)\in\mathcal{A}_{c_1,c_2,\delta_1,\delta_2}$, we have 
\[
\rho_n(m,p)\rightarrow\\exp\{-\phi(p)\};
\]
and
\[
\left|L_0(m,n)\right|=\log(mn)/\phi(p)+o_p(1),
\]
where $\phi(p)$ is a positive function and will be defined in (\ref{def:phi}).

Applying our theory to the multi-scale detection method in \cite{Aos06Filament}, we describe a multi-scale significant run algorithm that can reliably detect the concentration of data near a smooth curve, without knowing the smoothness information $\alpha$ or $\beta$ in advance, provided that the portion of points on the curve $\epsilon_N$ exceeds $T(\alpha,\beta)N^{1/(1+\alpha)}$. 
Our $T(\alpha,\beta)$ is smaller than that in \cite{Aos06Filament}, which indicates stronger detection ability using our theory. In the target tracking problem, our method provides a reliable threshold such that the false alarm probability vanishes very quickly as we get more and more sample points.


The rest of the paper is organized as follows. 
In Section \ref{chp:PseudoTreeModel}, we present a pseudo-tree model and study the critical probability and its reliability problems. 
In Section \ref{Chp:BernoulliNet}, we first  summarize the previous work of the Bernoulli net. 
Notice that from any node in the inflating net, there exists  a pseudo-tree, defined in Section \ref{chp:PseudoTreeModel}.
 Based on the results on the pseudo-tree model in Section \ref{chp:PseudoTreeModel}, we further provide the extensions to the Bernoulli net beyond the fixed number of rows. 
We present some potential applications of the longest run method in image detection problems in Section \ref{chp:Applications}. 
We conclude our work in Section \ref{sec:conclusion}. 
All proofs are relegated to Section \ref{sec:proofs}.

\section{Pseudo-Tree Model}
\label{chp:PseudoTreeModel}
In this section, we will first introduce the pseudo-tree model in Section \ref{pseudoTreeModel}. Then we provide our results on the critical probability and the asymptotic behaviors on the significant runs in pseudo-tree models in Section \ref{sec:cs}. Finally, we extend our model to the high-dimensional pesudo-tree model and generalize our results in the 2D case in Section \ref{extenPseuTree}.
\subsection{Model Introduction}
\label{pseudoTreeModel}
In this section, we first present a model which has some similarity to a \emph {regular} or \emph{complete-tree} model (\cite{FSP:EAC:2010,  trailInMaze}). 
Consider, for example, the lattice with nodes of the form
\begin{equation}\label{pseudoLattice}
V=\{(i,j)\in\mathbb{Z}^2: -iC\leq j\leq iC, i\geq 0\},
\end{equation}
and oriented edges $(i,j)\rightarrow(i+1, j+s)$, where $|s|\leq C$. 
We call $(0,0)$ the origin of the graph and sometimes use $0$ to denote the origin. Let $Y_{i,j}$ be the i.i.d. Bernoulli$(p)$ state variables corresponding to the node $(i,j)$. 
We say the node $(i,j)$ is significant, if $Y_{i,j}=1$, and insignificant if  $Y_{i,j}=0$.  In this paper, we are interested in the length of significant runs starting at the origin, which is a path consisting of only significant nodes in the graph. 
See Figure~\ref{Fi:pstree} for a sketch of the model.
\begin{figure}[hbt]
\centering\includegraphics[width = 95mm]{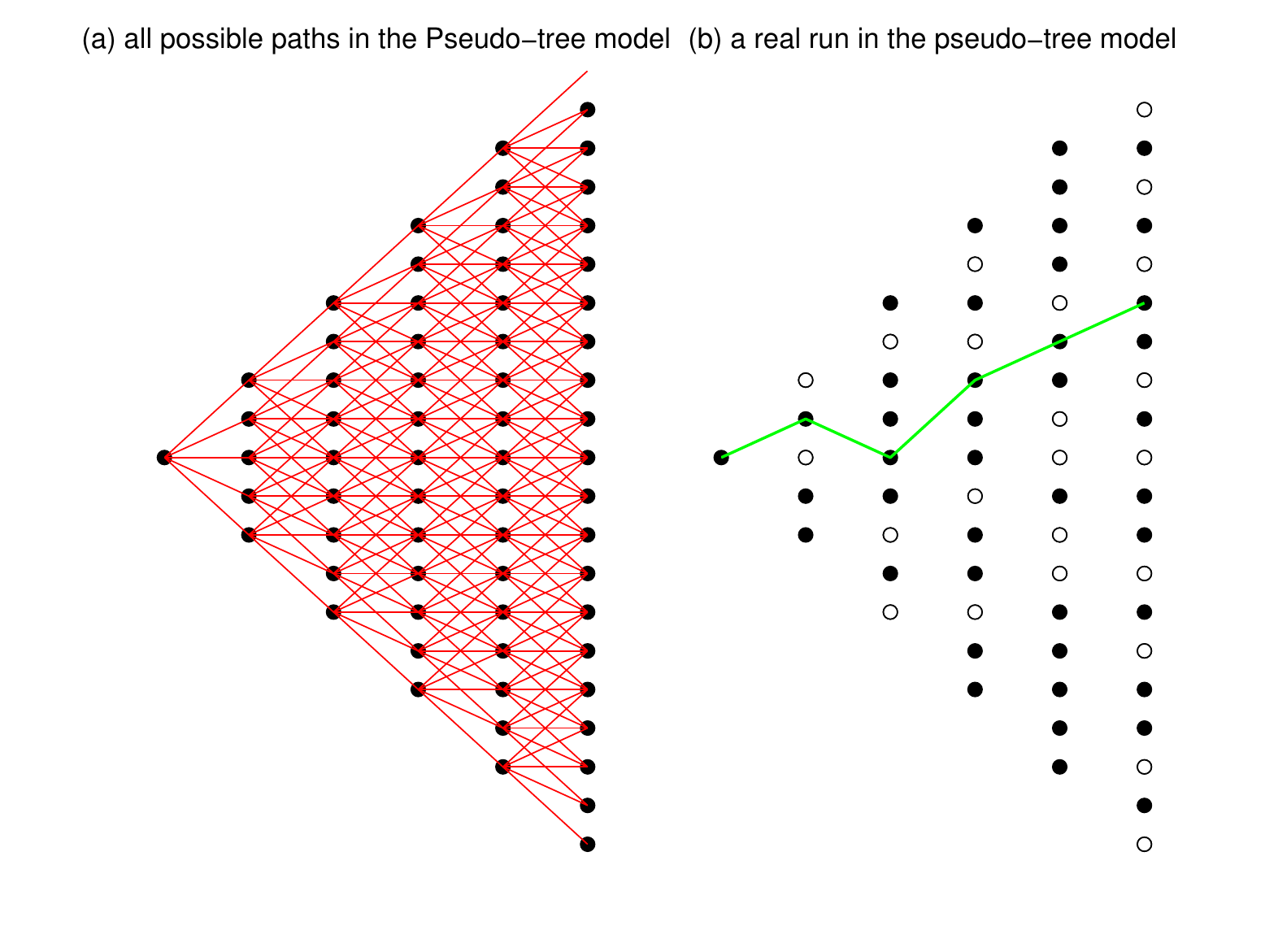}
\caption{A sketch of pseudo-tree model with the connectivity constraint $C=2$. (a) gives all the possible edges in the model. In (b) solid nodes are significant. The green path shows a possible real run in the pseudo-tree model }\label{Fi:pstree}. 
\end{figure}

Note that even though the number of runs of length $k$ in Pseudo-tree model and the regular tree model with $2C+1$ descendants are the same (both equal $(2C+1)^{k-1}$), the numbers of nodes are considerably different in the first $k$ columns---about $k^2C$ for the former and about $(2C+1)^k$ for the regular tree .

Let $p_c$ denote the critical probability for the site percolation in the Pseudo-tree model, defined as the supremum over all $p\in(0,1)$ such that the size of the significant run at the origin is finite with probability $1$, which is mathematically defined in Equation (\ref{eq:criticalProb}). 
By our knowledge, this model has not been fully studied yet and we will elaborate some results in the next section. 
Analogous to the model presented here, recent papers (\cite{trailInMaze, ClustDetecPerco})  have studied the oriented and non-oriented significant clusters or runs in a regular lattice.

\subsection{Results}
\label{sec:cs}

In this section, we give some results about the significant runs in Pseudo-tree model  $V$ presented in (\ref{pseudoLattice}). 
The difference between the \emph{Pseudo-tree} and \emph{Regular-tree} model is that the number of nodes in the former grows quadratically with the depth, as opposed to grow exponentially with the depth in the latter. 
Besides, in \emph{Pseudo-tree} model, different runs may share the same edges and therefore the behaviors of distinct runs here are quite correlated. 

\subsubsection{Notation}
\label{notation}
We shall introduce some notation. Observe that there is only one node in the $0$-th column, namely the origin $(0,0)$ and there are $2kC+1$ nodes in the $k$-th column, namely the nodes $(k, -kC),\ldots,(k,0),\ldots,(k,kC)$. For $k\in\mathbb{Z}^{+}$, let $B(k)=\{(k,-kC), \ldots, (k,0), \ldots, (k,kC)\}$ be the set of nodes in $k$-th column in $V$.

Let $\theta_{k}(p)$ denote the probability that $(0,0)$ is connectible to the $(k-1)$th column by a significant run, which implies, 
$$\theta_{k}(p)=\mathbb{P}_p((0,0)\leftrightarrow B(k-1)).$$ 
In other words, $\theta_{k}(p)$ is the probability that there is a significant run of length at least $k$ starting at the origin. 
Given any $x=(x_1,x_2)\in\mathbb{Z}^{2}$, let $\theta^x_{k}(p)$ be the probability that $x$ connects the $(x_1+k-1)$-th column with a significant chain. It is easy to see that $\theta^{x}_{k}(p)$ does not depend on the status of the nodes before the $x_1$-th column and $\theta^{x}_{k}(p)=\mathbb{P}_p(\{x\leftrightarrow B(x_1+k-1)\})=\theta_k(p)$. 
Because $\theta_k(p)$ only involves finitely many nodes, one can easily see that $\theta_k(p)$ is a continuous function of $p\in[0,1]$. Throughout the paper, we will  sometimes use $n$ as a subscript instead of $k$.

\subsubsection{Critical Probability}
\label{defOfTheta}
Given the above notations, we state some properties of the function $\theta_k(p)$ as follows:
\begin{itemize}
\item $\theta_{k_1}(p)\leq\theta_{k_2}(p)$, if $k_1\geq k_2$, which implies $\theta(p) \vcentcolon= \lim_{k\to\infty}\theta_{k}(p)$ exists;
\item $\theta_k(0)=0$ and $\theta_k(1)=1$, for any $k\geq1$, which implies $\theta(0)=0$ and $ \theta(1)=1$;
\item $\theta_{k}(p)$ and $\theta(p)$ are nondecreasing with respect to $p$.
\end{itemize}

Thus $\theta(p)$ is the probability that there is a significant run in $V$ starting from the origin and heading towards right forever when the probability of a node to be open is $p$. In light of this, we define $p_{c}$ to be the critical probability,
i.e., 
\begin{equation}
\label{eq:criticalProb}
p_c \vcentcolon= \sup\{p\in[0,1]:\theta(p)=0\}.
\end{equation}
So $p_c$ is the critical probability, above which it is possible to have an infinite significant run starting from any node in \emph{Pseudo-tree} model.

Recall that in the $r$-regular tree model, the critical probability $p_c=1/r$. Our first result shows that in the \emph{Pseudo-tree} model, the critical probability is no smaller than $1/r$, where $r=2C+1$ (See \cite{B:BOP}).
\begin{theorem}\label{T:p-c}
The critical probability of the Pseudo-tree model  $p_c \geq\frac{1}{2C+1}$.
\end{theorem}
In the beam-let model of \cite{Aos06Filament}, each node is connectible to $81$ nodes in the next column. 
Thus this theorem explains the reason that the authors there took the membership threshold $N^{\ast}$ such that $p=\mathbb{P}(\text{Poisson}(2)>N^{\ast})=\frac{p_0}{81}$ for some $p_0\in(0,1)$.

\subsubsection{Asymptotic rate of $\theta_k(p)$}

In this part, we show that under the sub-critical phase $p<p_c$
\[
\theta_{k}(p)=\mathbb{P}_p(0\leftrightarrow B(k-1))=O(k\exp\{-k\phi(p)\}),
\]
where $\phi(p)>0$ is a decreasing function of $p$.
\begin{theorem}\label{T:subexp}
Suppose $0<p\leq 1$. There exist positive constants $\sigma_1$ and $\sigma_2$, independent of $p$, and a unique function $\phi(p)$, such that
\begin{equation}\label{eq:subexp}
\sigma_1 k^{-1}\exp\{-k\phi(p)\}\leq\theta_{k}(p)\leq\sigma_2 k\exp\{-k\phi(p)\},
\end{equation}
for any $k\geq 1$. In particular,
\begin{equation}\label{def:phi}
\frac{\log\theta_{k}(p)}{k}\rightarrow -\phi(p) \mbox{ as } k \rightarrow \infty.
\end{equation}
\end{theorem}

The next corollary gives the limit of $\frac{\theta_k(p)}{\theta_{k-1}(p)}$.
\begin{corollary}\label{T:thetaratio}
\begin{equation}\label{lim:ratio}
\lim_{k\to\infty}\frac{\theta_{k}(p)}{\theta_{k-1}(p)}=\exp\{-\phi(p)\}.
\end{equation}
\end{corollary}

Given Theorem \ref{T:subexp}, one may speculate that $\phi(p)\rightarrow\infty$ as $p\rightarrow 0$, since $\theta(p)=0$ as $p=0$ and the theorem merits when $\phi(p)>0$. 
We will show $\phi(p)$ has the desired properties as $p<p_c$ in the following corollary.
\begin{corollary}\label{C:phi}
The function $\phi(p) \vcentcolon=  \lim_{k\to\infty}-\frac{\log\theta_{k}(p)}{k}$ have the following properties:
\begin{enumerate}
\item $\phi(p)$ is a continuous function on $(0,1]$;
\item $\phi(p)$ is strictly decreasing on $(0,p_c)$ and constantly $0$ when $p_{c}\leq p\leq1$;
\item $\lim_{p\to0}\phi(p)=\infty$.
\end{enumerate}
\end{corollary}

\begin{remark}
By observing Corollary \ref{C:phi}, Theorem \ref{T:subexp} is of no value when $p\geq p_{c}$ because $\phi(p)$ is constantly $0$ in the supercritical phase.
\end{remark}
Figure \ref{Fi:rho} gives the tendency of $-\frac{\log\theta_{k}(p)}{k}$ against $k$ for different values of $p$ when $C=1$.
\begin{figure}[hbt]
\centering\includegraphics[height=3.0in]{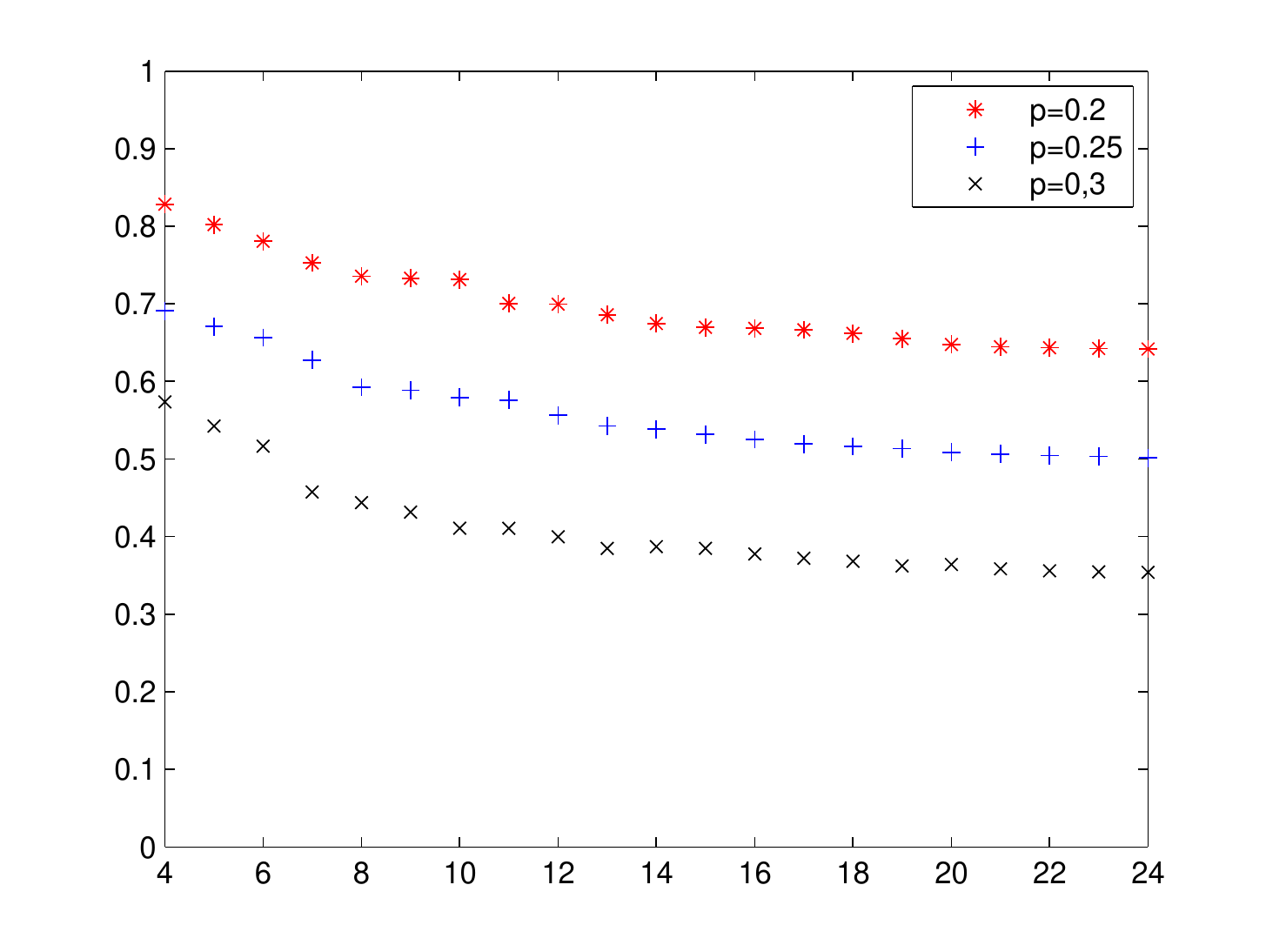}
\caption{A sketch of simulated result of $-\log\frac{\theta_k(p)}{k}$ against $k$ with $p$ being $0.2$, $0.25$, $0.3$ when $C=1$}\label{Fi:rho}
\end{figure}

\subsection{Extension to Pseudo-tree model in dimension $d'=d+1$}\label{extenPseuTree}
This section emphasizes that our results above for the \emph{Pseudo-Tree} model can be extended to other graphs and, in particular, to the analog of models  in higher dimensions.

The pseudo-tree model in dimension $d'=d+1$ is  the analogous lattice of (\ref{pseudoLattice}) in higher dimension
\begin{eqnarray*}
V^d=&&\{(i, j_1, \ldots, j_d)\in\mathbb{Z}^{d'}: -iC_k\leq j_k\leq iC_k, \\
&&k=1,\ldots, d, i\geq 0\},
\end{eqnarray*}

with oriented edges $(i,j_1,\ldots,j_d)\rightarrow(i+1,j_1+s_1,\ldots, j_d+s_d)$, where $|s_k|\leq C_k\in\mathbb{Z}^{+}$, $k=1,\ldots,d$. 
We denote $\theta_k^d(p)$ to be the probability that there is a significant run of length at least $k$ starting at the origin and $p_c^d$ to be the critical probability. 
We use the superscript $d$ to emphasize the notation in higher dimension. 

With these definitions of the graphs, we have the following results in higher dimension. 
The proofs of these theorems do not require any argument in addition to what we have already presented, and so they are omitted. 
\begin{theorem}\label{criticalProbInHighDim}
The critical probability of the forgoing pseudo-tree model in dimension $d\rq{}=d+1$ satisfies $p_c^d\geq\frac{1}{(2C_1+1)\times\ldots\times(2C_d+1)}$.
\end{theorem}
\begin{theorem}\label{phiInHighDim}
 For $0<p\leq 1$, there exist positive constants $\sigma_1^d$ and $\sigma_2^d$, independent of $p$, and there exists a unique function $\phi^d(p)$, which is strictly decreasing and positive when $p<p_c$; constantly $0$ otherwise, such that
\[
\sigma_1^dk^{-d}\exp\{-k\phi^d(p)\}\leq\theta_k^d(p)\leq\sigma_2^dk^d\exp\{-k\phi^d(p)\},
\]
for any $k\geq 1$. In particular, it follows that
\begin{equation}\label{phiInHigherDimension}
-\frac{\log\theta_k^d(p)}{k}\rightarrow\phi^d(p).
\end{equation}
\end{theorem}

More generally, let $\mathbb{Z}_{+}$ be the set of nonnegative integers. 
For any set $\mathcal{C}\subset\mathbb{Z}_{+}^d$, we may extend the condition of the oriented edges to a more general condition such as $(i, j_1,\ldots,j_d)\rightarrow(i+1,j_1+s_1,\ldots,j_d+s_d)$, where $(s_1,\ldots,s_d)\in\mathcal{C}$. 
It is straightforward to get the analogous results as above except that $p_c\geq 1/\text{card}\{\mathcal{C}\}$. Details are omitted here.

\section{Bernoulli Net}
\label{Chp:BernoulliNet}
In this section, we focus on studying the Bernoulli net in a two dimensional rectangular region, where both the number of rows and columns can go to infinity.
We first introduce the model in Subsection \ref{subsec:BernoulliModel}.
Then we review the previous results on Bernoulli, which mainly considers the scenario of fixed number of rows, in Subsection \ref{subsec:previousBernoulli}.
Our results on the asymptotic behaviors of the infinite Benoulli net are presented in Subsection \ref{subsec:inftyBernoulli},  \ref{subsec:longest}, and \ref{subsec:ext} on conditional across probability, rate of longest significant run, and extensions to higher dimensions, respectively.
\subsection{Model Introduction}\label{subsec:BernoulliModel}
We consider an $m$-by-$n$ array of nodes, in which there are $m$ rows and $n$ columns. Such an array can be considered as a grid in a two dimensional rectangular region, $([1,n]\times[1,m])\cap\mathbb{Z}^2$. 
Assume that each node with coordinate $(i,j), 1\leq i\leq n, 1\leq j\leq m$, is associated with a Bernoulli$(p)$ state variable $X_{i,j}$ i.e.,
\[
\mathbb{P}(X_{i,j}=1)=p=1-\mathbb{P}(X_{i,j}=0),
\]
where $p\in [0,1]$ is given. 
Assume state variables of nodes are i.i.d. If $X_{i,j}=1$, then the node is called significant (or open); otherwise, it is non-significant (or closed). 
Any two nodes in the grid, say $(i_1,j_1)$ and $(i_2,j_2)$ are \emph{connected} if and only if $|i_1-i_2|=1$ and $|j_1-j_2|\leq C$, with $C$ a prescribed positive integer. 
Define a chain of length $\ell$ as a chain of $\ell$ connected nodes, i.e.,
\begin{equation}\label{def:chain}
\begin{split}
&\{(i_1,j_1),(i_1+1,j_2),\ldots,(i_1+\ell-1,j_{\ell}): \\
&\left|j_k-j_{k-1}\right|\leq C, \forall k=2,\ldots,\ell\}.
\end{split}
\end{equation}
A \emph{significant (or open)} run refers to a chain with all the nodes being significant. 
We call such a system a \emph{Bernoulli net}. 
We are interested in the length of the longest significance run in this net. 
Throughout the paper, we denote the longest significant run in this net by $L_0(m,n)$ and its length by $\left|L_0(m,n)\right|$. 
Though in some papers \emph{runs},  \emph{chains} and \emph{clusters} have different definitions, here we treat them as synonyms. 
Such a model is used in the detection of filaments in a point cloud image (\cite{Aos06Filament, HuoandChen}) and networks of piecewise polynomial approximation (\cite{NetworksPoly}). 

Apparently, the length $\left|L_0(m,n)\right|$ depends on parameters $n$, $m$, $p$, and $C$. Figures \ref{Fi:ComparisonOfDependCP} and \ref{Fi:comparisonOfDependMN} give graphical representations of the relationships between the length $\left|L_0(m,n)\right|$ and parameters $C, p, m, n$. 
Number of simulations is $1,000$ for each histogram. The following presents a summary of the results.
\begin{itemize}
\item For fixed values of $m$ and $n$, when the value of $C$ or $p$ is increased, the distribution of $\left|L_0(m,n)\right|$ changes dramatically. These can be seen in Figure \ref{Fi:ComparisonOfDependCP}.
\item For fixed values of $C$ and $p$, if the value of $m$ or $n$ is doubled, the change of $\left|L_0(m,n)\right|$ is not significant. These can be seen in Figure \ref{Fi:comparisonOfDependMN}.
\end{itemize}
\begin{figure}[htbp]
\centering
\begin{tabular}{cc}
(a) & (b) \\
 \includegraphics[height =1.7in,width=1.73in]{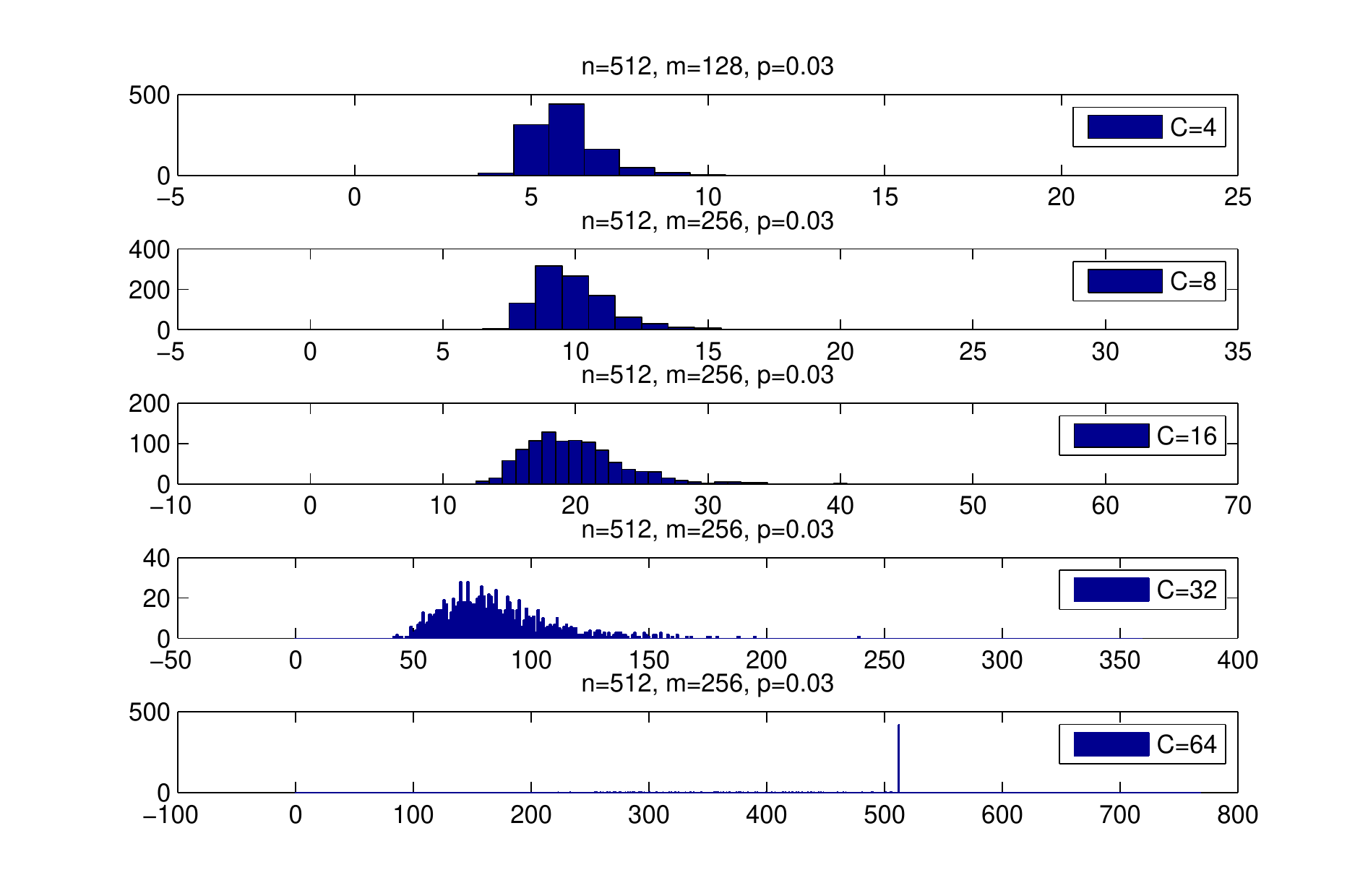} &
 \includegraphics[height =1.7in,width=1.73in]{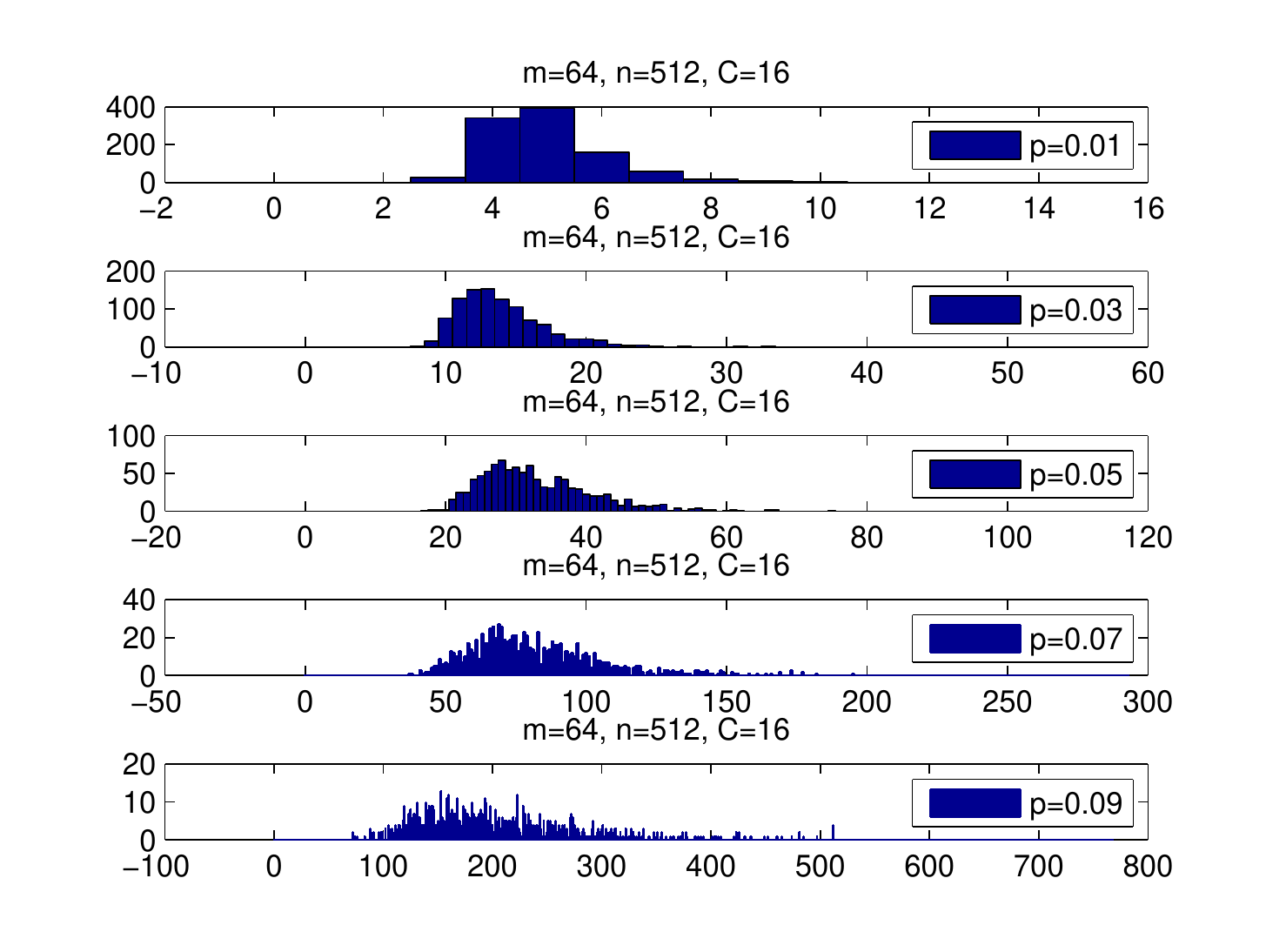}
\end{tabular}
\caption{(a) $\left|L_0(m,n)\right|$ versus $C$: effects of connectivity. Every time when the value of $C$ is doubled, the histogram of $\left|L_0(m,n)\right|$ is shifted to the right significantly.
(b) $\left|L_0(m,n)\right|$ versus $p$: effects of significance probability $p$. When the value of $p$ is increased, the histogram of $\left|L_0(m,n)\right|$ is shifted to the right.}\label{Fi:ComparisonOfDependCP}
\end{figure}

\begin{figure}[htbp]
\centering
\begin{tabular}{cc}
(c) & (d) \\
\includegraphics[height=1.8in, width=1.8in]{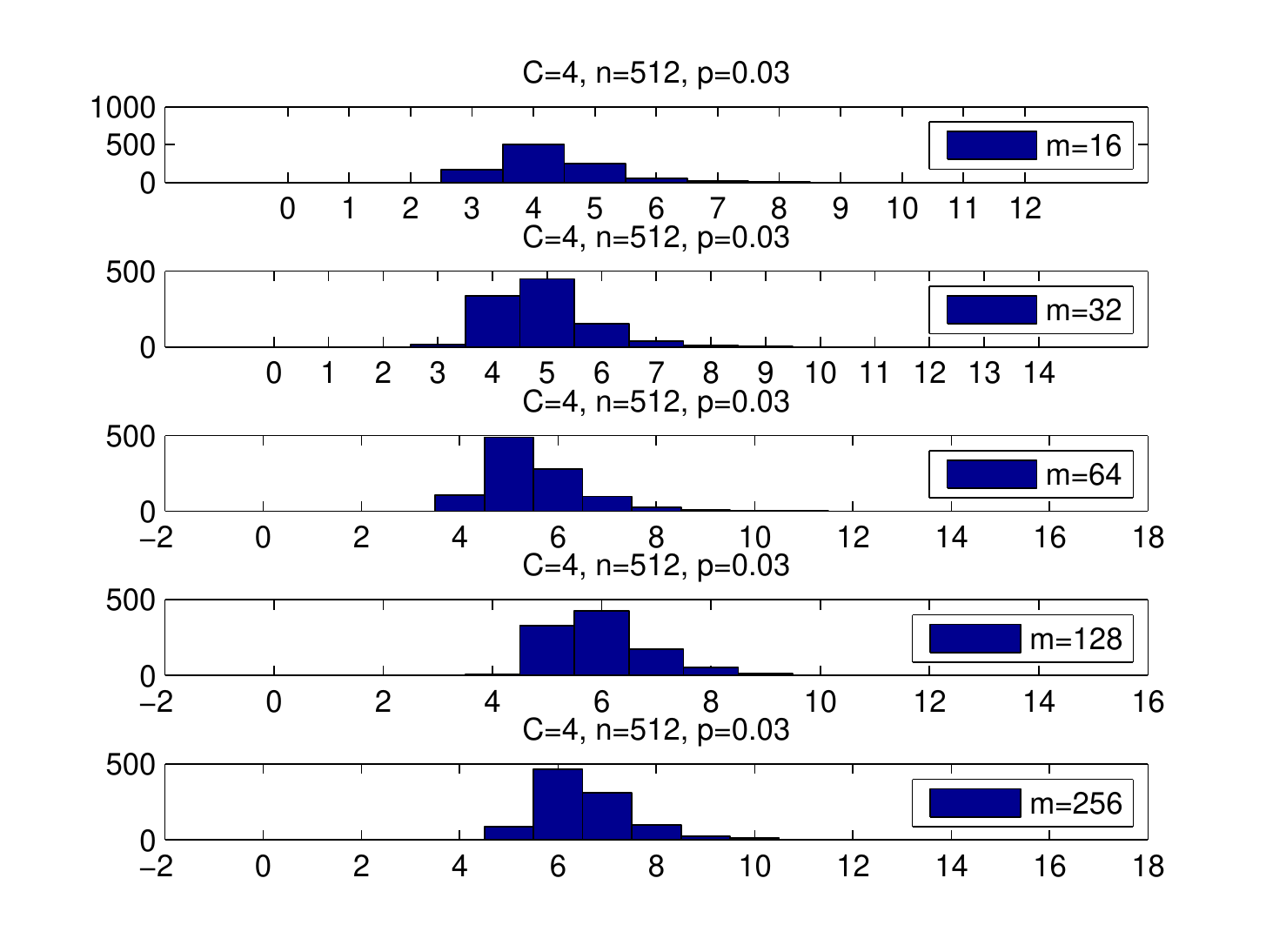} &
\includegraphics[height=1.8in, width=1.8in]{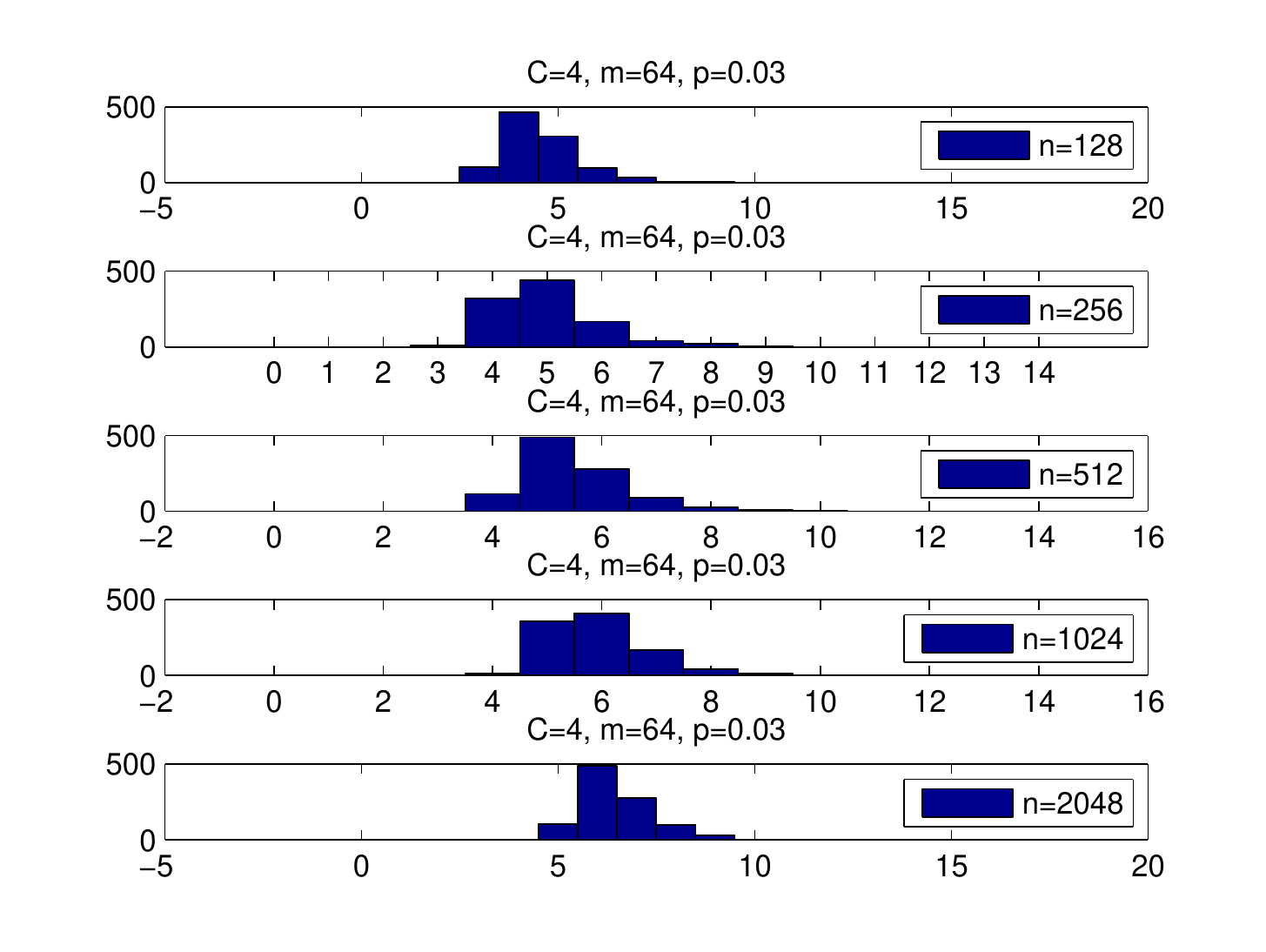}
\end{tabular}
\caption{(c) $\left|L_0(m,n)\right|$ versus $m$: effects of heights. When the value of $m$ is doubled, the histogram of $\left|L_0(m,n)\right|$ does not change dramatically.
(d) $\left|L_0(m,n)\right|$ versus $n$: effects of the width of the Bernoulli net. Every time when the value of $n$ is doubled, the histogram of $\left|L_0(m,n)\right|$ does not change dramatically.} \label{Fi:comparisonOfDependMN}
\end{figure}

\subsection{A thin slab}
\subsubsection{Previous Work}\label{subsec:previousBernoulli}
\label{subsec:PreviousWork}
In this section, we discuss the previous work related to the model in \cite{JASA06LongSigRun}, which focuses on the scenario where the number of rows $m$ is fixed.  
We will discuss the relationship between $\phi(p)$ mentioned in (\ref{def:phi}) and the conditional across probability defined in \cite{JASA06LongSigRun}. 
We list the results in \cite{JASA06LongSigRun}. For proofs of these results, please refer to \cite{JASA06LongSigRun} and references therein.

The first result is motivated by reliability-focused work \cite{Reliab}.
\begin{theorem}\label{th:stab}
Let $P_{k}(m,p)=\mathbb{P}_{C,p}(\left|L_{0}(m,k)\right|=k)$ denote the probability that the length of the longest significant run is $k$, when there are exactly $k$ columns and $m$ rows. We have
\begin{equation}\label{eq:stabni}
\begin{split}
&(1-P_{k}(m,p))^{n-k+1}\leq\mathbb{P}_{C,p}(\left|L_{0}(m,n)\right|<k)\\
\leq&[1-q^{m}P_{k}(m,p)]^{n-k+1},
\end{split}
\end{equation}
where $q=1-p$.
\end{theorem}

The following lemma introduces a constant $\rho(m,p)$ depending on $m$ and $p$, which is important in the asymptotic distribution of $\left|L_{0}(m,n)\right|$. 
\begin{lemma}\label{L:rho}
Define $\rho_{k}(m,p)=\frac{P_{k}(m,p)}{P_{k-1}(m,p)}$. There exists a constant $\rho(m,p)$ in $(0,1)$ that depends on $m, C$, and $p$, but not on $k$ such that
\[
\lim_{k\to\infty}\rho_{k}(m,p)=\rho(m,p).
\]
\end{lemma}
Let an \emph{across} be a significant run that passes all columns from left to right. 
The ratio $\rho_{k}(m,p)$ is the conditional probability that conditioning on the fact that there is an across in the previous $\left(k-1\right)$ columns, there will be an across for $k$ columns. 
We may call this the \emph{chance of preserving across significant runs} or \emph{conditional across probability}. 
The foregoing lemma shows that as the number of columns goes to infinity, the chance of preserving across significant runs converges to a constant.

Now we will recall the result in \cite{JASA06LongSigRun}, which is a generalization of the well-known Erd$\ddot{o}$s-R$\acute{e}$nyi law (See \cite{LofLHR, PetrovVV, NLawofLN}), which is equivalent to the following theorem for $m=1$, since $\rho(1,p)=p$.
\begin{theorem}\label{th:er-re}
For any fixed $m\in\mathbb{N}$, as $n\rightarrow\infty$, we have
\[
\frac{\left|L_{0}(m,n)\right|}{\log_{1/\rho(m,p)}n}\rightarrow1, \text{\quad almost surely}.
\]
\end{theorem}
 Given this theorem, it is easy to obtain the following result, which states the relation of $\rho$ and $(m, p)$. 
 Since $\left|L_{0}(m,n)\right|$ actually depends on $p$, we use the notation $\left|L_{0}(m,n,p)\right|$ in the next corollary to make the dependence explicit.
\begin{corollary}\label{co:depofrho}
Given a pair of positive integers $m_1,m_2$ and a pair of probabilities $p_1, p_2$ with $m_1\leq m_2$ and $p_1\leq p_2$, we have
\[
\rho(m_1,p_1)\leq\rho(m_2,p_1)\text{\quad and \quad}\rho(m_1,p_1)\leq \rho(m_1,p_2)
\]
\end{corollary}

Let us recall the result which states the asymptotic distribution of $\left|L_{0}(m,n)\right|$, the proof of which employs the Chen-Stein approximation method. See \cite{JASA06LongSigRun} and \cite{TERLinD}.
\begin{theorem}\label{th:asymdist}
There exists a constant $A_1>0$, that depends only on $m, C$, and $p$ but not on $n$, such that for any fixed $t$, as $n\rightarrow\infty$, we have
\[
\mathbb{P}_{p}(\left|L_{0}(m,n)\right|<\log_{1/\rho(m,p)}n+t)\rightarrow\exp\{-A_1\cdot\rho(m,p)^{t}\}.
\]
\end{theorem}
The analogous result for a one-dimensional Bernoulli sequence is well known. 
See \cite{OELDA}. The foregoing theorems provide a comprehensive description on the asymptotic distribution of the length of the longest significant run $\left|L_{0}(m,n)\right|$ in a Bernoulli net when the row number $m$ of the array is fixed.

\subsubsection{Asymptotic behavior of conditional across probability}\label{subsec:inftyBernoulli}
We see that all the results in the last subsection depend on $\rho(m,p)$. 
If $\rho(m,p)\rightarrow 1$ as $m\rightarrow\infty$, then Theorems \ref{th:er-re} and \ref{th:asymdist} may not hold. 
We shall next discuss the asymptotic behavior of $\rho_{k}(m,p)$.
 
Recall that $\theta(p)$ is the probability that there exisits an infinite significant chain rooted at the origin and $p_c=\sup\{p\in[0,1], \theta(p)=0\}$. 
We first consider a special case in the array with $m=\infty$ and $n=\infty$. 
In the following, if $m=\infty$, we employ the lattice of $([1,n]\times\mathbb{Z})\cap\mathbb{Z}^2$ rather than $([1,n]\times[1,\infty])\cap\mathbb{Z}^2$. 
This theorem indicates that as $(m,n)\rightarrow(\infty,\infty)$, the behavior of the length of the longest significant run will be quite different in the cases that $p>p_c$ and $p<p_c$.
\begin{theorem}\label{T:kolmo}
Let an array have $\mathbb{Z}^{+}\times\mathbb{Z}$ nodes, where $\mathbb{Z}^{+}$ denotes the set of all nonnegative integers. The probability that there exists an infinite significant chain (when the marginal probability of a node to be open equal to $p$), denoted by $\mu(p)$, in the lattice satisfies
\[
\mu(p)=\begin{cases}
0, \quad &\text{if}\quad p<p_c,\\
1, \quad &\text{if}\quad p>p_c.
\end{cases}
\]
\end{theorem}

We next separate our discussion into the super-critical phase, where  $p>p_c$ and the sub-critical phase, where $p<p_c$.\\
\\
\textbf{Phase $p>p_c$:} 
Our first result shows that in the phase that $p>p_c$, $\rho(m,p)\rightarrow 1$ as $m\rightarrow\infty$ for any $p>p_c$.
\begin{theorem}\label{th:RhoLimSup}
For any $p>p_c$, we have
\begin{equation}\label{eq:RhoLimSup}
\lim_{m\to\infty}\rho(m,p)=\rho(\infty,p)=1
\end{equation}
where $\rho(\infty,p)=\lim_{k\to\infty}\rho_{k}(\infty,p)=\lim_{k\to\infty}\frac{P_{k}(\infty,p)}{P_{k-1}(\infty,p)}$, and $\rho_k(\infty, p)$ is the conditional probability that there is an across in the first $k$ columns conditioned on the event that there is an across in the first $k-1$ columns when there are infinitely many rows.
\end{theorem}

We note that $\lim_{m\to\infty}\rho(m,p)=1$ in the case of $p>p_c$. Recall that we introduce $\phi(p)$ and its property in Corollary \ref{C:phi}. $\phi(p)\equiv0$ on  $p\in[p_c,1]$. So we have the iterated limit
\begin{equation}\label{eq:rholim}
\lim_{m\to\infty}\lim_{k\to\infty}\rho_k(m,p)=\exp\{-\phi(p)\},
\end{equation}
 when $p\in[p_c,1]$. Recall that in (\ref{cond:m}), we define
 \[
 \mathcal{A}_{c_1,c_2,\delta_1, \delta_2}=\{(m,n): c_1n^{1+\delta_1}\leq m\leq c_2\exp[n(\phi(p)-\delta_2)]\},
 \]
for positive $c_1$, $c_2$, $\delta_1$ and $\delta_2$. 
In the following, we use $\rho_n(m,p)$ instead of $\rho_k(m,p)$ and we will show below the double limit of $\rho_n(m,p)$ is $\exp\{-\phi(p)\}$, when $p<p_c$ as $n\rightarrow\infty$, $m\rightarrow\infty$ and $(m,n)\in\mathcal{A}_{c_1,c_2,\delta_1,\delta_2}$ by Chen-Stein's approximation method (See \cite{twommts}).\\
\\
\textbf{Phase $p<p_c$:}
Recall that in Theorem \ref{T:subexp}, we introduce $\theta_{n}(p)$, which is the probability that there is a significant run of size $n$ connecting the origin and $B(n-1)$. 
In Theorem \ref{th:stab} we introduce $P_{n}(m,p)$, which is the probability that the length of the longest significant run is $n$ when there are exactly $n$ columns. 
To determine the limit of $\rho_{n}(m,p)=\frac{P_{n}(m,p)}{P_{n-1}(m,p)}$, we need to know $P_{n}(m,p)$ when both $n$ and $m$ are very large positive integers.

\begin{theorem}\label{th:ptoinfty}
Let $([1,n]\times[1,m])\cap\mathbb{Z}^2$ be the integer lattice with the probability of nodes being open equal to $p$. 
Let $P_{n}(m,p)$ be the probability of the event that there is a significant run from the first column to the last column of the lattice, which is called an across run (or across) in Lemma \ref{L:rho}. Then if $p<p_c$, we have
\[
P_{n}(m,p)=1-\exp\{-m\theta_n(p)\}+o(1),
\]
as $m\rightarrow\infty, n\rightarrow\infty$ and $(m,n)\in\mathcal{A}_{c_1,c_2,\delta_1,\delta_2}$.
In particular,
we have $\rho_n(m,p)\rightarrow\exp\{-\phi(p)\}$ as $m\rightarrow\infty, n\rightarrow\infty$ and $(m,n)\in\mathcal{A}_{c_1,c_2,\delta_1,\delta_2}$.
\end{theorem}

In \cite{JASA06LongSigRun}, the authors provide a method to calculates the values of $\rho(m,p)$ (see Table \ref{tb:rho}), when $m$ is small and fixed by finding out the solution of $\pi=\pi P$, where $P$ is a transition matrix. 
See also (11) in \cite{JASA06LongSigRun}.
\begin{table}[htbp]
\caption{The values of $\rho$ for different values of $m$ and $p$,
when $C=1$.}
\begin{center}
\begin{tabular}{|c|c|c|c|c|c|c|}
\hline
p&0.1&0.2&0.3&0.4&0.5 & 0.6\\
\hline
m=4&0.2444 &   0.4564& 0.6341 & 0.7758& 0.8804 & 0.9482\\
m=8& 0.2654 & 0.4955& 0.6869 & 0.8363 & 0.9383& 0.9876\\
m=10&0.2691 & 0.5022& 0.6958& 0.8467& 0.9486& 0.9930\\
\hline
\end{tabular}
\end{center}
\label{tb:rho}
\end{table}

One can use simulation to find $\phi(p)$ in the case of $p<p_c$ and thus get some idea about $\rho(m,p)$ as $m$ becomes sufficiently large. 
See Figure ~\ref{Fi:rho}. 
The simulation below is done for the length of the longest significant chain in \cite{JASA06LongSigRun} for $n=64$, $m=128$, $C=3$ and $p=0.05$ when nodes are assumed to be independent. 
See Figure~\ref{Fi:freq}. The result is based on $10,000$ simulations.
\begin{figure}[htbp]
\centering
\begin{tabular}{c}
(a)\\  
 \includegraphics[height=2.4in]{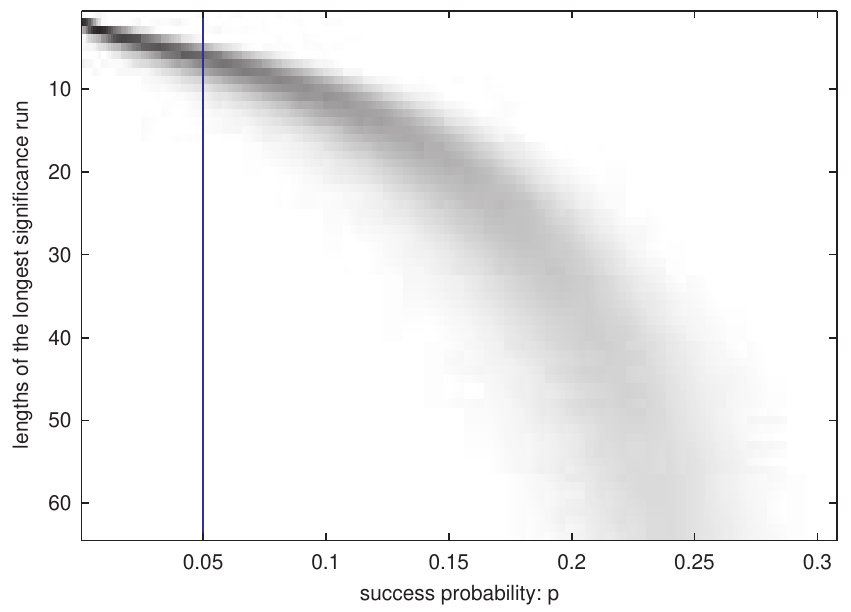} \\
 (b) \\
 \includegraphics[height=2.4in]{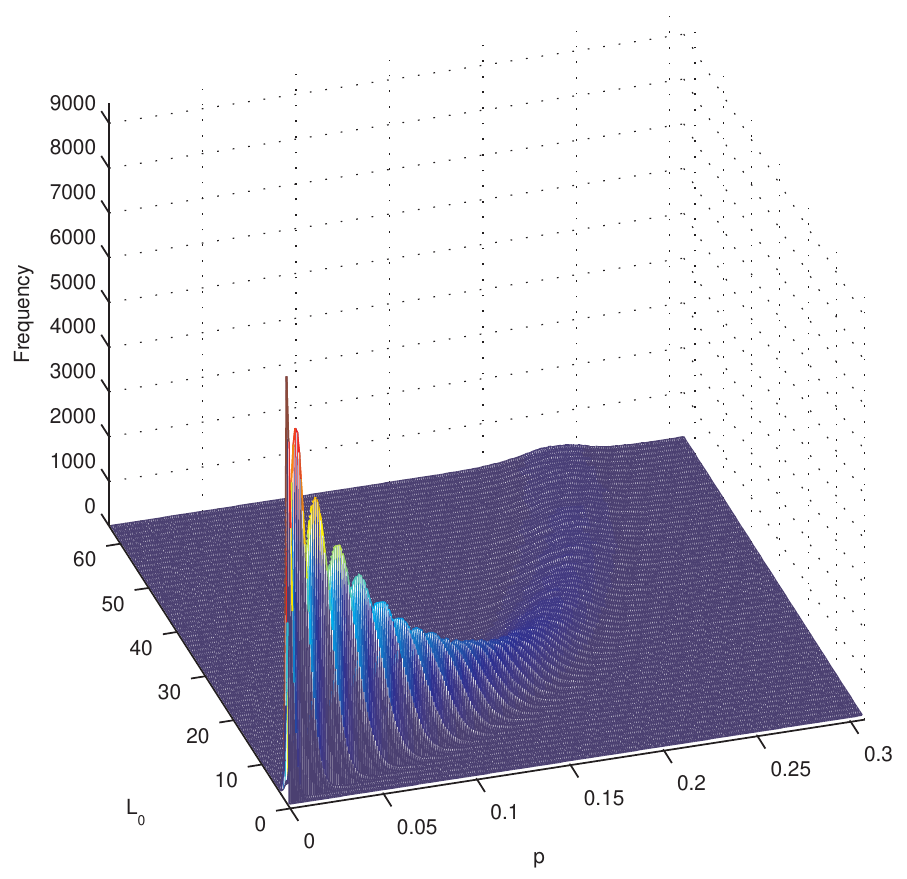} \\
(c)  \\
\includegraphics[height=2.5in]{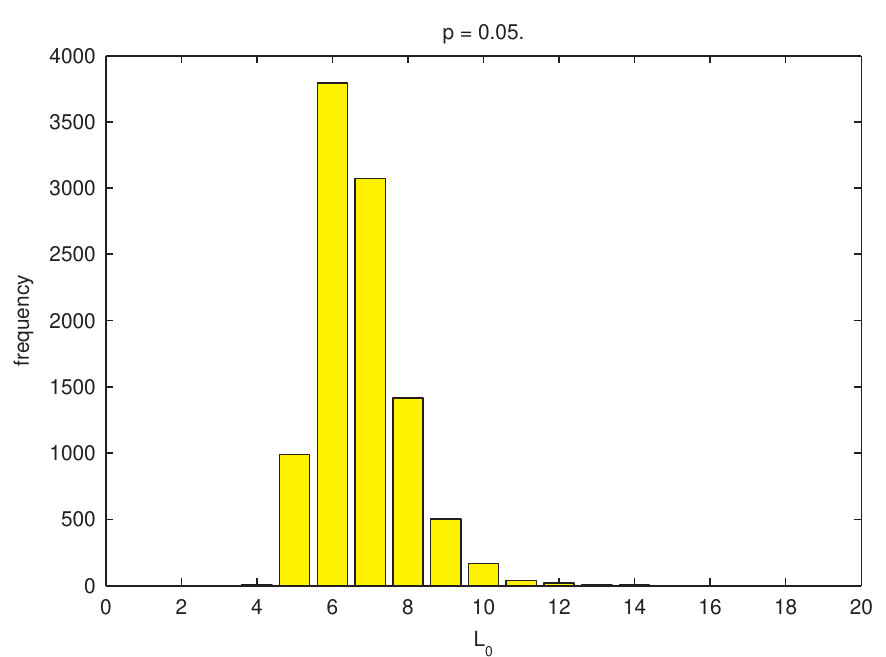} 
\end{tabular}
\caption{(a) An image plot, the distribution of $\left|L_0(m,n)\right|$ (under
$n=64, m=128, C=3$) as a function of $p$ ($0 < p < 0.3075$). The
intensity of the image is proportional to the frequency of
$\left|L_0(m,n)\right|$ (which is specified by the y-coordinate) given a value of
$p$ (which is the x-coordinate) out of $10,000$ simulations. (b) A
mesh plot of the same data as in (a). (c) For $p=0.05$, the
histogram of $L_0$ based on the same $10,000$ simulations. Note
this can be viewed as one vertical slice from (a), or similarly a
slice from (b).} \label{Fi:freq}
\end{figure}

\subsection{Rate of the longest significant run}\label{subsec:longest}

The following is an extension of Theorem 2 in \cite{JASA06LongSigRun} in the case that the Bernoulli net enlarges as $m\rightarrow\infty$and $n\rightarrow\infty$. 
In the following, $\log$ denotes the logarithm with base $e$ unless the base is explicitly specified.
\begin{theorem}\label{th:mainconv}
When $p<p_c$, then as $m\rightarrow\infty$ and $n\rightarrow\infty$, we have that
\begin{equation}\label{mainconveq}
\frac{\left|L_0(m,n)\right|}{\log(mn)}\rightarrow \frac{1}{\phi(p)}, \text{\quad in probability},
\end{equation}
where $\phi(p)$ is a strictly decreasing, continuous function defined in (\ref{def:phi}), which is positive in $(0,p_c)$ and constantly $0$ otherwise.
\end{theorem}
From Theorem \ref{th:mainconv}, it is apparent that asymptotically $m$ and $n$ do not have a significant impact on the length of the longest significant run $\left|L_0(m,n)\right|$. 
We showed that the critical probability $p_c>\frac{1}{2C+1}$ and $\left|L_0(m,n)\right|$ will have significantly different asymptotic behaviors between the case $p<p_c$ and $p>p_c$. 
Therefore, as $C$ and $p$ increases, $\left|L_0(m,n)\right|$ will increase dramatically while the increment of $m$ and $n$ do not have a significant impact on the length $\left|L_0(m,n)\right|$. 
Figure \ref{Fi:ComparisonOfDependCP} and \ref{Fi:comparisonOfDependMN} support this argument. 

\subsection{Extension}\label{subsec:ext}
 This section emphasizes that our results above can be extended to the case of models  in higher dimensions.
\begin{itemize}
\item\emph{Inflating Bernoulli net in dimension $d'=d+1$.} This is the graph with nodes $([1,n]\times[1,m_1]\times\ldots\times[1,m_d])\cap\mathbb{Z}^{d'}$. 
Assume that each node with coordinate $(i,j_1,\ldots,j_d)$, $1\leq i\leq n$, $1\leq j_k\leq m_k$, $k=1,\ldots,d$ is associated with a Bernoulli(p) random variables, where $p\in[0,1]$ is given. 
Equip this graph with oriented edges $(i,j_1,\ldots,j_d)\rightarrow(i+1, j_1+s_1,\ldots, j_d+s_d)$, where $s_k=1,\ldots,C_k$, $k=1\ldots, d$ for prescribed $C_k\in\mathbb{Z}^{+}$. 
We say a chain to be significant if all the nodes along the chain are significant and denote $L_0(n,m_1,\ldots,m_d)$ to be the longest significant run in this model with length $\left|L_0(n,m_1,\ldots,m_d)\right|$.
\end{itemize}

By Theorem \ref{th:mainconv}, it is easy to see that we have the following asymptotic rate of the longest significant run.
\begin{theorem}\label{asymRateInHighDim}
Let $\phi^d(p)$, defined in (\ref{phiInHigherDimension}), be the higher dimensional version of $\phi(p)$. As $n\rightarrow\infty$, $m_1\rightarrow\infty, \ldots, m_d\rightarrow\infty$, we have that
\begin{equation}\label{th:mainconvHighDim}
\frac{\left|L_0(n,m_1,\ldots,m_d)\right|}{\log(nm_1\ldots m_d)}\rightarrow\frac{1}{\phi^d(p)} \text{\quad in probability},
\end{equation}
\end{theorem}

\section{Applications}
\label{chp:Applications}
In this section, we are going to see some applications of the above theory in hypothesis testing problems.
In Section \ref{subsec:algorithm}, we first introduce the dynamic programming (DP) algorithm to find the longest significant run in an image.
In Section \ref{subsec:egBernoulli}, we use the example of detecting an anomalous run in a Bernoulli net to illustrate our theory on constructing asymptotically powerful test.
In Section \ref{subsec:egMultiscale}, we consider the multi-scale detection of filamentary structure. 
We first review the results in the literature and then we apply our theory on longest run to solve this problem.
The last application of target tracking problems is shown in Section \ref{subsec:egTarget}. 
We propose to apply our longest run theory to detect potential target.
We show that our method provides a reliable threshold such that the false alarm probability vanishes very quickly as we get more and more sample points.

Through this section, let $L_0(n,m)$ and $\left|L_0(n,m)\right|$ denote the longest significant run and  the length of the longest significant run in $([1,n]\times[1,m])\cap\mathbb{Z}^2$, respectively. 

\subsection{Dynamic programming algorithm finding $\left|L_0(n,m)\right|$}\label{subsec:algorithm}
For a node $(i,j)\in ([1,n]\times[1,m])\cap\mathbb{Z}^2$, we use $z(i,j)=1$ $(=0)$ to denote the significance (insignificance) of node $(i,j)$. 
When $X_{i,j}\sim\text{Bernoulli}(p_0)$, we have $z(i,j) = x(i,j)$.
Given a realization $\{X(i,j): 1\leq i\leq m, 1\leq j\leq n\}$, let $Y_1$ be an array $\{Y_1(i,j): 1\leq i\leq m, 1\leq j\leq n\}$, such that
\begin{eqnarray*}
Y_1(i,1)&=&z(i,1), \text{for } i=1,\ldots, m;\\
Y_1(i,j)&=&z(i,j)[1+\max_{i'\in\Omega(i)}Y_1(i', j-1)], \\
          &&\text{for } i=1,\ldots, m, j=2,\ldots,n,
\end{eqnarray*}
where $\Omega(i)=\{i': \left|i'-i\right|\leq C, 1\leq i'\leq m\}$ denotes the set containing neighboring indices of $i$. Finally, the value  $\left|L_0(n)\right|$ can be computed as follows:
\[
\max_{(i,j)\in\mathcal{S}}Y_1(i,j).
\]
It is not hard to see that this algorithm takes $Cmn$ time for $C>0$.

\subsection{Detection of an anomalous run in a Bernoulli net}\label{subsec:egBernoulli}

In this subsection, we consider the problem of detecting an anomalous run in Bernoulli net. 
For simplicity, we only state the low dimension case i.e., $([1,n]\times[1,m])\cap\mathbb{Z}^2$. Let $\mathcal{L}(n,m)$ be a class of chains in $([1,n]\times[1,m])\cap\mathbb{Z}^2$, where a chain is defined as a subset of nodes which is connected as in (\ref{def:chain}). 
Under the null hypothesis, each node $(i,j)$ is i.i.d. associated with a random variable $X_{i,j}$, which has Bernoulli distribution with parameter $p_0$,  i.e.,
\[
\mathbb{H}_0(n,m): X_{i,j}\sim\text{Bernoulli}(p_0), i.i.d., \forall (i,j).
\]
Under the alternative hypothesis, where there exists an unknown  chain $L\in\mathcal{L}(n,m)$, and the variables with index in $L$ have a Bernoulli distribution with parameter $p_1>p_0$, i.e.,
\begin{eqnarray*}
\mathbb{H}_{1}(n,m): &&X_{i,j}\sim\text{Bernoulli}(p_1), \forall (i,j)\in L; \\
&&X_{i,j}\sim\text{Bernoulli}(p_0), \forall (i,j)\not\in L,\\
&& \mbox{for some unknown } L. 
\end{eqnarray*}
Denote the length of the anomalous chain $L$ by $\left|L\right|$. 
For this detection problem, we may consider the test based on the longest significant run in the Bernoulli net $([1,n]\times[1,m])\cap\mathbb{Z}^2$. 
By Erd$\ddot{o}$s-R$\acute{e}$nyi law (\cite{NLawofLN}), the longest significant run in $L$ almost surely has length $\log_{1/p_1}\left|L\right|$ as $|L|\rightarrow\infty$. Thus if 
\begin{equation}\label{cond:sep}
\log_{1/p_1}\left|L\right|>\log(nm)/\phi(p),
\end{equation}
then the two hypotheses can be separated significantly. 
Let $T$ be such a test, if 
\[
\left|L_0(n,m)\right|>\log(nm)/\phi(p),
\]
then we reject $\mathbb{H}_0(n,m)$; otherwise accept $\mathbb{H}_0(n,m)$. 

For a test $T$ , if $T=1$ , we reject $\mathbb{H}_0$  and accept $\mathbb{H}_0$ otherwise; then if 
\begin{equation}\label{cond:assym0}
\mathbb{P}(T=0\big|\mathbb{H}_1)+\mathbb{P}(T=1\big|\mathbb{H}_0)\rightarrow0,
\end{equation} $T$ is called asymptotically powerful test in \cite{IEEE} and this criterion (\ref{cond:assym0}) is widely used in cluster detection literatures (See for example \cite{trailInMaze, NetworksPoly, ClustDetecPerco, AnomaClusterInNetwork}). 
\begin{theorem}
Under the condition (\ref{cond:sep}), the test $T$, which is based on the length of the longest significant run, is an asymptotically powerful test.
\end{theorem}
\begin{proof}
If (\ref{cond:sep}) holds, then by Theorem \ref{th:mainconv}, it is easy to see that
\begin{equation}\label{assym0}
\begin{split}
&\mathbb{P}(\left|L_0(n,m)\right|>\log(nm)/\phi(p)\big|\mathbb{H}_0(n,m))\\
+&\mathbb{P}(\left|L_0(n,m)\right|
\leq\log(nm)/\phi(p)\big|\mathbb{H}_{1,L}(n,m))\\
\rightarrow &0,
\end{split}
\end{equation}
as $(n,m)\rightarrow(\infty,\infty)$.
\end{proof}

In general, this detection problem can be extended to an exponential model, for instance, the following detection problem in the model with normal distribution,
\begin{eqnarray*}
\mathbb{H}^N_0(n,m): &&X_{i,j}\sim N(0,1), i.i.d., \forall(i,j);\\
&&versus\\
\mathbb{H}^N_{1}(n,m): &&X_{i,j}\sim N(\mu,1), \forall (i,j)\in L; \\
&&X_{i,j}\sim N(0,1), \forall (i,j)\not\in L,\\
&&\mbox{for some unknown } L \mbox{ and } \mu >0.
\end{eqnarray*}

After thresholding the values at each node, it is equivalent to the detection problem in the Bernoulli net. 
We are going to discuss this problem in our future work. 
The test based on the length of the longest significant chain has also been considered in \cite{ClustDetecPerco, perim1, perim4}. 

\subsection{Multi-scale detection of filamentary structure}\label{subsec:egMultiscale}
In this section, we will revisit the problem of multi-scale detection of filamentary structure. 
This has been studied in \cite{Aos06Filament}, which we review in Section \ref{subsec:review}.
We then revisit the problem and apply our proposed theory to it in Section \ref{subsec:revisit}.
\subsubsection{Background}
\label{subsec:review}
To be self-contained, we will recall the problem of the length of the longest significant run proposed in \cite{Aos06Filament}, where the authors present a detection method for some filamentary structure in a background of uniform random points. 
Suppose we have $N$ data points $X_i\in[0,1]^2$, which at first glance seem to be uniformly distributed in the unit square. 
Here, for $1<\alpha\leq2$, we define that $\text{H$\ddot{o}$lder}(\alpha,\beta)$ is the class of functions $g:[0,1]\rightarrow[0,1]$ with continuous derivative $g\rq{}$ that obeys
\[
\left|g\rq{}(x)-g\rq{}(y)\right|\leq\alpha\beta\left|x-y\right|^{\alpha-1}.
\]
Consider the problem of testing
\begin{eqnarray*}
\mathbb{H}_0: &&X_i\stackrel{\mbox{i.i.d.}}{\sim}\text{Uniform}(0,1)^2,\\
versus\\
\mathbb{H}_1(\alpha,\beta):&&X_i\stackrel{\mbox{i.i.d.}}{\sim}(1-\epsilon_N)\text{Uniform}(0,1)^2\\
&&+\epsilon_N\text{Uniform}(\text{graph}(f)),\\
&&\mbox{with unknown } f\in\text{H$\ddot{o}$lder}(\alpha,\beta),
\end{eqnarray*}
where $\text{graph}(f)$ is the graph of the function $f$ within the area $[0,1]^2$. In other words, for the problem of testing, we believe that a relatively small fraction $\epsilon_n$ of points lie on a smooth curve in the plane.

In \cite{JASA06LongSigRun}, the detection model mentioned in \cite{Aos06Filament} is partially considered and the authors present the convergence rate and the asymptotic distribution of the longest significant run on a Bernoulli Net. 
However, the row number of the model in \cite{JASA06LongSigRun} is fixed, while in \cite{Aos06Filament} the vertical size of the model is increasing very fast when the number of random points tends to infinity. 
Besides, the nodes in \cite{JASA06LongSigRun} are assumed to be independent while in \cite{Aos06Filament} the nodes are only associated. See \cite{AORV}.

We will review the model in \cite{Aos06Filament} first. 
Suppose we have $N$ random points uniformly distributed in the square $[0,1]\times[0,1]$. 
In particular, we use $J=\lceil \log_2(N)\rceil$ to denote its dyadic logarithm. 
The variable $j$ will index dyadic scales $2^{-j}$ and will range over $0\le j\le J$. 
We fix a positive integer $S>1$ to control the maximum of $|\text{slope}|$ we will be able to detect.

Let $R(j,k,\ell_1,\ell_2)$ be a parallelogram with vertical sides that is $\omega=2^{-j}$ wide by $t=2^{-(J-j)+1}$ high, where $j$ runs through our set of scale indices $\{0,\ldots,J\}$. 
The regions in question have a midline that bisects them vertically and will be tilted at a variety of angles. 
And notice that these regions are highly anisotropic.

The parameters $k$ and $\ell_{i}, i=1,2$, control the horizontal location of the regions and the vertical location and the slope of the midline. 
There is an underlying assumption that we are only interested in regions whose major axis has a slope bounded in absolute value by $S$. 

To get a vivid impression of this model, see Figure~\ref{Fi:para} and Figure~\ref{Fi:square} below. Let $\delta_1=\frac{t}{4}$ and $\delta_2=\frac{t}{4\omega}$ (these depend implicitly on $j$ and $N$). 
The parallelogram $R(j,k,\ell_1,\ell_2)$ will be centered at $c=((k+\frac{1}{2})\omega,\ell_1\delta_1)$ and its middle line will have slope $s=\ell_2\delta_2$. Here $0\le k<\omega^{-1}$, $\ell_1$ runs through the set $0,\ldots,\delta^{-1}_1-1$ and $\ell_2$ runs through the set
$-S\delta^{-1}_2,\ldots,0,\ldots, S\delta^{-1}_2$. 
We gather all such regions at level (scale) $j$ in
$\mathcal{R}(j)=\{R(j,k,\ell_1,\ell_2):k,\ell_1,\ell_2\}$ and therefore we have
$2^j\times 2^{J-j+1}\times S\cdot2^{J-2j+2}+1$ or $O(N^2)$ parallelograms in total. 
To organize the regions, we define
a directed graph $\mathcal{G}(j)=(\mathcal{V}(j),\mathcal{E}(j))$, with vertices $\mathcal{V}(j)$ and edges $\mathcal{E}(j)$.

\begin{figure}[htbp]
\centering
\includegraphics[height=2.5in]{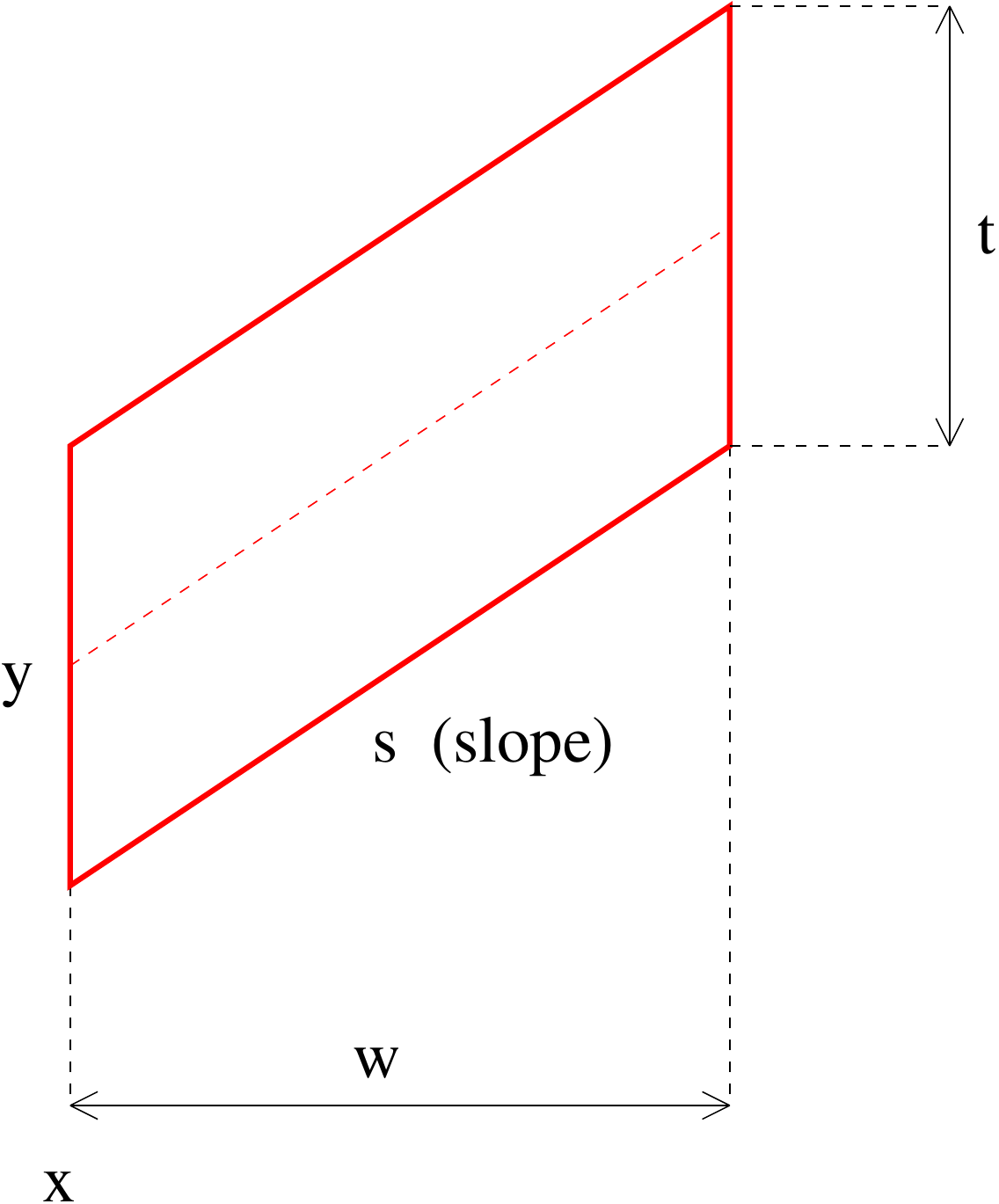}
\caption{An Anisotropic `Strip' $R$}
\label{Fi:para}
\end{figure}

\begin{figure}[htbp]
\centering
\includegraphics[height=2.9in]{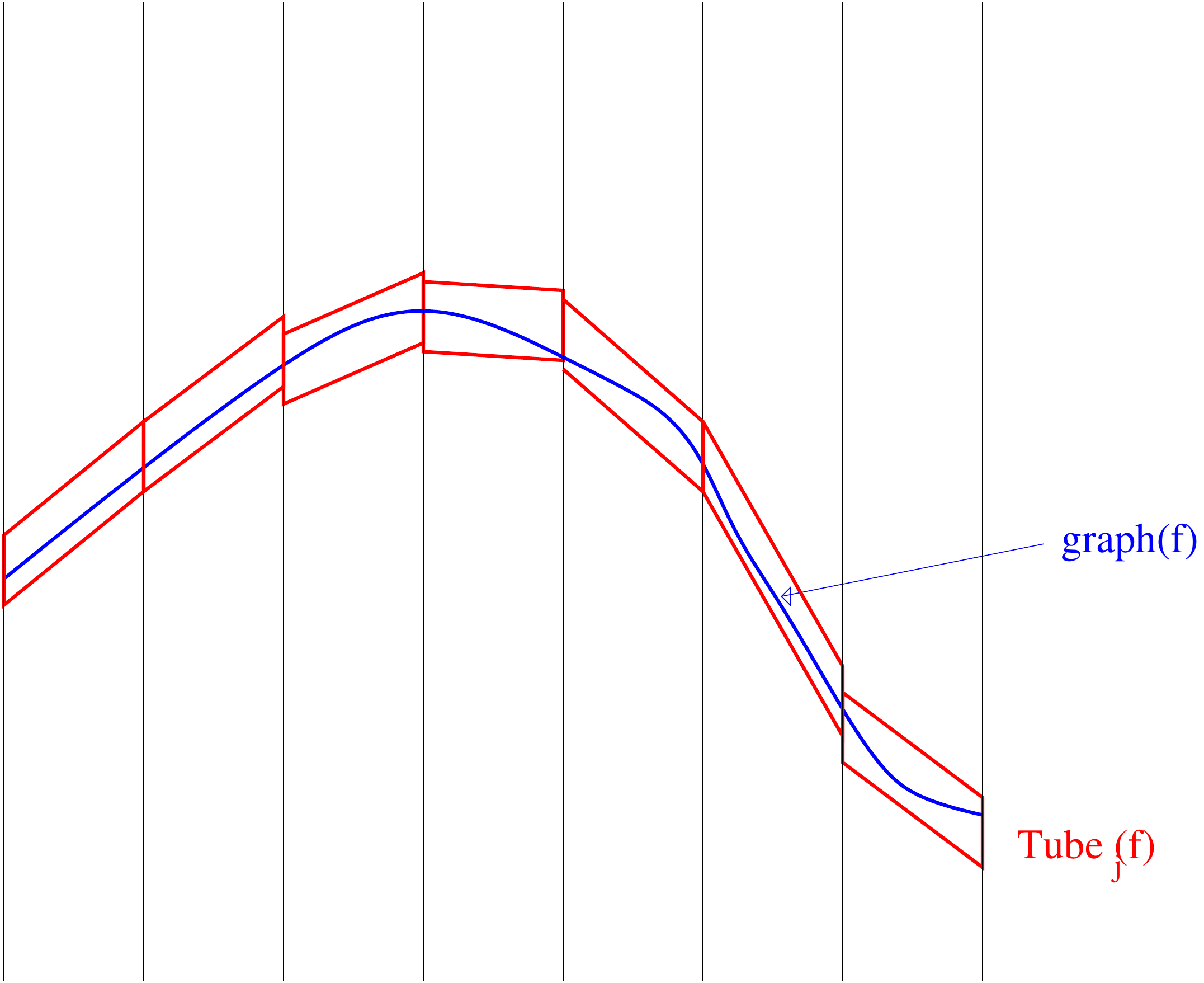}
\caption{$graph(f)$ (in blue) covered by $Tube_j(f)$ (in red).}
\label{Fi:square}
\end{figure}

The vertices are simply the regions $\mathcal{R}(j)$, i.e., $\mathcal{V}(j)\equiv\mathcal{R}(j)$. 
The edges connect regions by good continuation, namely, to regions that are horizontally adjacent, and that have altitudes and slopes that are nearly the same-less than $t$ and $\frac{t}{\omega}$ apart, respectively. 
Formally, we have the directed edges in $\mathcal{E}(j)$ as
\begin{equation}\label{connectionConstraint}
(k,\ell_1,\ell_2)\rightarrow(k+1,\ell_1+\ell_2+u,\ell_2+v),
\end{equation}
where $|u|\le 4, |v|\le 4$ and we call (\ref{connectionConstraint}) the connectivity of edges. 
The mapping between these discrete parameters is intended to insure that the regions pack together horizontally and that they are fairly closely spaced in both vertical position and slope.

For every region $R\in\mathcal{R}(j)$, we count the number of the points that fall into $R$, denoted by $N(R,j)$.
We define a significance indicator, which is nonzero when the counts $N(R,j)$ exceeds a prescribed threshold $N^{\ast}$, i.e., 
\begin{equation}\label{eq:mem}
s(R)=\textbf{1}_{\{N(R,j)>N^{\ast}\}}.
\end{equation}
We say that $N^\ast$ is the counting threshold in the following. The significance indicator may be viewed as a label on the regions $R$, producing a sequence of a labeled graphs
\[
\Sigma(j)=(\mathcal{V}(j),\mathcal{E}(j),\sigma(j)),
\]
where $\sigma(j)=(s(R))$ gives the labels on $R\in\mathcal{R}(j)$. We call this the $j$-th significance graph.

In each significance graph, we employ a depth-first search algorithm to explore all significance paths
\[
\pi=(R_1,R_2,\ldots,R_m),
\]
that is, sequence of vertices that are:
\begin{itemize}
\item[(a)] all significant, $s(R_k)=1$;
\item[(b)] all connected, $(R_k,R_{k+1})\in\mathcal{E}(j)$.
\end{itemize}
We record the maximum path length in each significant graph as follows:
\[
\left|L^{\max}_{N,j}\right|=\max\{\text{length}(\pi): \pi \text{\quad is a significant path in\quad}\Sigma(j)\},
\]
\[
\left|L^{\max}_N\right|=\max_{j} \left|L^{\max}_{N,j}\right|.
\]
The decision of the hypothesis testing problem is that we compare $\left|L^{\max}_N\right|$ with a length threshold: 
If $\left|L^{\max}_N\right|\leq \left|L^{\ast}_N\right|$, accept $\mathbb{H}_0$; if $\left|L^{\max}_N\right|>\left|L^{\ast}_N\right|$, then reject $\mathbb{H}_0$.

We call $\left|L^{\ast}_N\right|$ the decision threshold in the following. 
Under the assumption that $N$ points are randomly distributed in the square $[0,1]\times[0,1]$, the counting threshold determines the probability of $\{s(R)=1\}$. Because the area of each region is $\frac{2}{N}$, we have the following,
\[
\mathbb{P}(s(R)=1)=\mathbb{P}(\text{Bin}(N,\frac{2}{N})>N^{\ast}),
\]
where $\text{Bin}(N,\frac{2}{N})$ denotes the random variable with Binomial distribution of parameters $N$ and
$\frac{2}{N}$. Because $\text{Poisson}(2)$ is an approximation of $\text{Bin}(N,\frac{2}{N})$ when $N$ is sufficiently
large, we use $\text{Poisson}(2)$ instead in the following. 
The main result of this multi-scale detection method in \cite{Aos06Filament} is the following:
\begin{theorem}\label{mainResultOfAos06}
There is a single choice of threshold $N^{\ast}$ and $\left|L^{\ast}_N\right|$ so that for every $\alpha\in(1,2]$ and $\beta>0$, there is $T_{\ast}(\alpha,\beta,S)$ such that for each $\epsilon_{N}>T_{\ast}N^{-\alpha/(1+\alpha)}$
\[
\mathbb{P}(\text{test rejects } \mathbb{H}_0\big|\mathbb{H}_1(\alpha,\beta,S))\rightarrow 1, \quad\text{as}\quad n\rightarrow\infty,
\]
and at the same time
\[
\mathbb{P}(\text{test rejects } \mathbb{H}_0\big|\mathbb{H}_0)\rightarrow 0,\quad\text{as}\quad n\rightarrow\infty.
\]
\end{theorem}
\begin{remark}
We give the specifications of the foregoing thresholds. 
\begin{enumerate}
\item \textbf{$N^{\ast}$}: In \cite{Aos06Filament}, the authors define $N^{\ast}$ such that
\begin{equation}\label{def:N}
p=\mathbb{P}(\text{Poisson}(2)>N^{\ast})<\frac{p_0}{81},
\end{equation}
where $p_0\in(0,1)$ is some chosen number. 
\item \textbf{$ \left|L^{\ast}_N\right|$}: With the help of Erd$\ddot{o}$s-R$\acute{e}nyi$ law, the authors define the decision threshold
\begin{equation}\label{decisionThreshold}
\left|L^{\ast}_N\right|\equiv3\log_{1/p_0}N.
\end{equation}
\item \textbf{$T_{\ast}(\alpha, \beta, S)$}: The specification of $T_{\ast}(\alpha, \beta, S)$ in Theorem \ref{mainResultOfAos06} is a little bit complicated. First define 
\begin{equation}\label{def:pAstO}
p^{\ast}=p_0^{\frac{1}{18}},
\end{equation}
Let $\lambda^{\ast}$ be a constant that satisfies
\[
\mathbb{P}(\text{Poisson}(\lambda^{\ast})<N^{\ast})\leq \frac{1-p\ast}{2}.
\]
Then $T_{\ast}(\alpha, \beta, S)=2\lambda^{\ast}\beta^{\frac{1}{1+\alpha}}\sqrt{1+S^2}$. See \cite{Aos06Filament} for more details.

\end{enumerate}
\end{remark}
\subsubsection{A revisit using the theory of longest chain}
\label{subsec:revisit}
In this part, we will apply our theory to the model in \cite{Aos06Filament} for the detection problem.
Consider the problem of testing
\begin{eqnarray*}
\mathbb{H}_0: &&X_i\stackrel{\mbox{i.i.d.}}{\sim}\text{Uniform}(0,1)^2,\\
versus\\
\mathbb{H}_1(\alpha,\beta):&&X_i\stackrel{\mbox{i.i.d.}}{\sim}(1-\epsilon_N)\text{Uniform}(0,1)^2\\
&&+\epsilon_N\text{Uniform}(\text{graph}(f)),\\
&&\mbox{with unknown } f\in\text{H$\ddot{o}$lder}(\alpha,\beta),
\end{eqnarray*}
where  $N$ is the total number of points in $[0,1]\times[0,1]$ and $\epsilon_N>T_{\ast}N^{-\frac{\alpha}{1+\alpha}}$ is the portion of the points lying on the graph of the function. 

We can see that when the number of random points $N$ in $[0,1]\times[0,1]$ goes to infinity, the background of uniform random points can be treated as sampled from a (spatial) Poisson process. One of the properties of this Poisson process is that for any subregion $\Omega$  in the unit square, the number of points in this region, denoted by $N(\Omega)$, also has the Poisson distribution with parameter $N\cdot\left|\Omega\right|$, where $\left|\Omega\right|$ is the area of $\Omega$, i.e., $N(\Omega)\sim\text{Poi}(N\cdot|\Omega|)$. 
Another property of the (spatial) Poisson process is that for any two non-overlapping regions $\Omega_1$ and $\Omega_2$ in the unit square, the number of points in $\Omega_1$ and the number of points in $\Omega_2$, i.e., $N(\Omega_1)$ and $N(\Omega_2)$ respectively, are independent.

Let us now rephrase the main results of \cite{AORV} here.
\begin{definition}\label{def:assoc}
Let $T_1,T_2,\ldots,T_n$ be $n$ associated random variables if $\text{Cov}[f(\textbf{T}),g(\textbf{T})]\geq 0$, where $\textbf{T}=(T_1,T_2,\ldots,T_n)$, for all nondecreasing functions f and g, for which the expectations $\mathbb{E}(f), \mathbb{E}(g)$ and $\mathbb{E}(fg)$ exist.
\end{definition}
It is known that
\begin{itemize}
\item any subset of associated random variables are associated;
\item nondecreasing functions of associated random variables are associated;
\item independent random variables are associated
\item let $x_1,\ldots,x_n$ be associated binary random variables, then
\[
\mathbb{P}(x_1=s,\ldots,x_n=s)\geq\mathbb{P}(x_1=s)\ldots\mathbb{P}(x_n=s),
\]
where $s$ can be either $0$ or $1$.
\end{itemize}

Regarding the parallelograms defined in Section \ref{subsec:review}, we have the following lemma.
\begin{lemma}
As the number of points in the unit square tends to infinity, The number of random points in two parallelograms $R_1$ and $R_2$ are associated. Furthermore, they are independent if $R_1$ and $R_2$ are non-overlapping.
\end{lemma}
\begin{remark}
This lemma can be simply proved by considering $E[(f(R_1,R_2) - f(R'_1,R'_2))(g(R_1,R_2) - g(R'_1,R'_2))] \geq 0$, where $f,g$ are two nondecreasing functions of $R_1,R_2$, and $R'_1,R'_2$ are i.i.d. copy of $R_1,R_2$.
\end{remark}
Indeed, if we use $s(R)$ to denote the state (significant or non-significant) of parallelogram $R$, then we have
\begin{equation}\label{ineq:signState}
\mathbb{P}(s(R_1)=a, s(R_2)=a)\geq\mathbb{P}(s(R_1)=a)\mathbb{P}(s(R_2)=a),
\end{equation}
where $a$ is either $0$ or $1$. The equality in (\ref{ineq:signState}) holds when $R_1$ does not overlap with $R_2$.

For a multi-scale detection problem, we construct an array of nodes in
\begin{equation}\label{def:V}
\mathcal{V}\equiv[1,2^j]\times[1,2^{J-j+1}]\times[-S2^{J-2j+1},S2^{J-2j+1}]\cap\mathbb{Z}^3,
\end{equation}
where $J=\lceil\log_2(N)\rceil$ and $0\leq j\leq J$. For any nodes 
\[
(k,\ell_1,\ell_2)\in[1,2^j]\times[1,2^{J-j+1}]\times[-S2^{J-2j+1},S2^{J-2j+1}]\cap\mathbb{Z}^3
\]
in the array, the three components represent the location index, the altitude index and the slope index, respectively. In light of the nodes in two dimension, we might consider $m=2^{J-j+1}\times(2S\cdot2^{J-2j+1}+1)$ nodes in the same strip as a column and thus there are $n=2^j$ columns in total. For any node $(k,\ell_1,\ell_2)$, it can be connected to
\begin{equation}\label{eq:connfash}
\begin{split}
&(k+1,\ell_1+\ell_2+u,\ell_2+v)\\
\in&[1,2^j]\times[1,2^{J-j+1}]\times[-S2^{J-2j+1},S2^{J-2j+1}]\cap\mathbb{Z}^3,
\end{split}
\end{equation}
where $|u|\leq 4$, $|v|\leq 4$. Each node is associated with a parallelogram in the algorithm mentioned in \cite{Aos06Filament} and therefore it is open with probability
\[
p=\mathbb{P}(N(R)>N^{\ast})\rightarrow\mathbb{P}(\text{Poisson}(2)>N^{\ast}), \text{\quad as\quad} N\rightarrow\infty,
\]
where $N(R)$ is the number of points in the parallelogram $R$ and $N^{\ast}$ is a counting threshold to be specified later. Due to the structure of the model in \cite{Aos06Filament}, the nodes in different columns are independent and all the nodes here are associated as $N\rightarrow\infty$. 

Consider the Pseudo-tree model in dimension $3$, as in Section \ref{extenPseuTree},
\[
V^2=\{(i,j_1,j_2)\in\mathbb{Z}^2: -4i\leq j_1\leq 4i, -4i\leq j_2\leq 4i, i\geq 0\},
\]
with oriented edges $(i,j_1,j_2)\rightarrow(i+1,j_1+s_1,j_2+s_2)$, where $\left|s_1\right|\leq 4$ and $\left|s_2\right|\leq 4$. We denote $\theta^2_k(p)$ to be the probability that there is a significant run of length at least $k$ starting at the origin and $p^2_c$ to be the critical probability. 
Revisiting the proofs of Theorems \ref{T:p-c}, \ref{T:subexp} and \ref{th:mainconv} together with their generalized results in Theorems \ref{criticalProbInHighDim}, \ref{phiInHighDim} and \ref{asymRateInHighDim}, we find that these results do not depend on the independence of nodes in the same column. 
The condition that nodes are associated in the same strip is sufficient for these theorems. 
By Theorem \ref{phiInHighDim}, there exist positive constants $\sigma^2_1$ and $\sigma^2_d$, independent of $p$, and there exists a unique function $\phi^2(p)$, which is strictly decreasing and positive when $p<p_c$; constantly $0$ otherwise, such that 
\[
\sigma^2_1 k^{-2}\exp\{-k\phi^2(p)\}\leq\theta^2_k(p)\leq\sigma^2_2 k^2\exp\{-k\phi^2(p)\} 
\]
 for any $k\geq 1$. In particular, it follows that 
\[
-\frac{\log\theta^2_k(p)}{k}\rightarrow\phi^2(p).
\]
Since each node in the array can be connected with at most $81$ nodes in the next column and hence $p_c\geq\frac{1}{81}$ by Theorem \ref{criticalProbInHighDim}.

Though in \cite{Aos06Filament}, the authors consider all scales in $\{j:0\leq j\leq J$ for $J=\lceil\log_{2}N\rceil\}$, we will consider $\{j:0\leq j\leq \lceil\frac{J+\log_2\beta}{1+\alpha}\rceil\}$ only. 
We shall point out here that the restriction on $j$ is a fairly reasonable assumption for the following reasons. First notice that if we choose $j>\lceil\frac{J+\log_2\beta}{1+\alpha}\rceil$, then the range of the slope index $[-S2^{J-2j+1},S2^{J-2j+1}]$ will be fairly small. 
Hence the parallelograms will be almost horizontal rectangles. Moreover, under $\mathbb{H}_1(\alpha,\beta)$, for scales $j\leq\lceil\frac{J+\log_2\beta}{1+\alpha}\rceil$, the parallelograms in the same column will be more overlapping which yields more significant nodes and hence the longer length of the significant runs. 
And it is easier to separate the null hypothesis $\mathbb{H}_0$ from the alternative hypotheses $\mathbb{H}_1(\alpha,\beta)$. 
The most important reason is that, in \cite{Aos06Filament}, the authors point out that under $\mathbb{H}_1(\alpha,\beta)$, there is some scalar $j^{\ast}$ such that the graph of the function is completely covered by a tube of parallelograms in this scale like the case in Figure \ref{Fi:square}. We call this \textit{containing tube $T_{j^{\ast}}(f)$}. 
It is shown that $j^{\ast}=\lceil\frac{J+\log_2\beta}{\alpha+1}\rceil$ (See Lemma 2.1-2.3 and their proofs in \cite{Aos06Filament}). In other words, using only scalars $j\leq\lceil\frac{J+\log_2\beta}{1+\alpha}\rceil$ is enough to cover the graph hence detect the filamentary structure under $\mathbb{H}_1(\alpha,\beta)$. 
Thus, it actually can save work to consider only the scales no larger than $\lceil\frac{J+\log_2\beta}{1+\alpha}\rceil$ without loss of generality. 
In case that $\alpha\in(1,2]$ and $\beta>0$ are unknown, it is possible to use $0.5001J$ instead of $\lceil\frac{J+\log_2\beta}{\alpha+1}\rceil$ for the reason that $\lceil\frac{J+\log_2\beta}{\alpha+1}\rceil\leq 0.5001J$ as $J\rightarrow\infty$. 
Denote $\lceil\frac{J+\log_2\beta}{\alpha+1}\rceil$ by $c_J$, which is the scale under which the whole graph of the function is guaranteed to be in a series of parallelograms, as shown in Figure \ref{Fi:square}.

Now we specify the asymptotic thresholds for our purpose. 
These thresholds are better and more intuitive than those in \cite{Aos06Filament}. 
Specifically, applying our theory, we set the threshold parameters as follows.

\begin{itemize}
\item Let the membership threshold $N^{\ast}$ satisfy the following property:
\[
p_0=\mathbb{P}(\text{Poisson}(2)>N^{\ast})<\frac{1}{81}\leq p_c.
\]
Here we can take $N^{\ast}=6$ so that $p_{0}\approx0.0045338<\frac{1}{81}$. 
\item Let the decision threshold $\left|L^{\ast}_N\right|$ be
\[
(1+\delta_3)\frac{2J\log 2}{\phi(p_0)},
\]
for some small $\delta_3>0$.
\item Define $p^{\ast}$ to be
\begin{equation}\label{def:pAst}
\exp\{-\phi(p_0)\frac{c_J(1-\delta_3)}{2J(1+\delta_3)}\}.
\end{equation}
\item We choose $\lambda^{\ast}$ such that
\begin{equation*}
\mathbb{P}(\text{Poisson}(\lambda^{\ast})>N^{\ast})>p^{\ast}.
\end{equation*}
\item Finally, we define $T_{\ast}(\alpha,\beta, S)$ to be $2\lambda^{\ast}\beta^{\frac{1}{1+\alpha}}\sqrt{1+S^2}$. 
\begin{equation}\label{ineq:lambda}
\mathbb{P}(N(R)>N^{\ast})>p^{\ast},
\end{equation}
where $N(R)$ is the number of points in the parallelogram $R$.
\end{itemize}

Let $\left|L^{\max}_N\right|$ denote the length of the longest significant run in Lattice $\mathcal{V}$, defined in (\ref{def:V}). 
Note that there are $4S2^{2J-2j}+2^{J+1}$ nodes in $\mathcal{V}$.
We will have the following results.
\begin{corollary}
 Under $\mathbb{H}_0$, by Theorem \ref{asymRateInHighDim}, for any small $\epsilon$ and $\delta_3>0$, with probability at least $1-\epsilon$, we have
\begin{equation}\label{errorTypeI}
\begin{split}
\left|L^{\max}_N\right|\leq&(1+\delta_3/2)\frac{\log(4S2^{2J-2j}+2^{J+1})}{\phi(p_0)}\\
\leq&(1+\delta_3)\frac{2J\log2}{\phi(p_0)},
\end{split}
\end{equation}
as $N\rightarrow\infty$.
\end{corollary}
\begin{corollary}
Under $\mathbb{H}_1(\alpha,\beta)$, with probability at least $1-\epsilon$ and for large $N$, by the Erd$\ddot{o}$s- R$\acute{e}$nyi Law and the Egoroff's Theorem (\cite{Egoroff1988}), we have that the length of the significant run in the tube $T_{j^{\ast}}(f)$ containing the function $f$, denoted by $\left|L_{j^{\ast}}(f)\right|$, satisfies
\begin{eqnarray}
\left|L_{j^{\ast}}(f)\right|&>&(1-\delta_3)c_J\log_{1/p^{\ast}}2\label{ineq:first}\\
&=&(1+\delta_3)\frac{2J\log2}{\phi(p_0)}\label{ineq:second}\\
&=&\left|L^{\max}_{N}\right|\label{ineq:third},
\end{eqnarray}
as $N\rightarrow\infty$. 
\end{corollary}
The inequality (\ref{ineq:first}) is due to the fact that (\ref{ineq:lambda}) holds for each parallelogram in the containing tube $T_{j^{\ast}}(f)$. 
The equality (\ref{ineq:second}) is due to the definition of $p^{\ast}$ in (\ref{def:pAst}). 

We thus can get the following conclusion:
\begin{theorem}
When $\epsilon_N>T_{\ast}(\alpha,\beta, S)N^{-\frac{\alpha}{1+\alpha}}$, by (\ref{errorTypeI}) to (\ref{ineq:third}), the test based on the length of the longest significant chain is asymptotically powerful, i.e.,
\begin{eqnarray*}
\mathbb{P}(\left|L_N^{\max}\right|>\left|L_{N}^{\ast}\right|\big|\mathbb{H}_0)&\rightarrow&0, \text{\quad as\quad} N\rightarrow\infty,\label{typeIerror}\\
\mathbb{P}(\left|L_N^{\max}\right|<\left|L_{N}^{\ast}\right|\big|\mathbb{H}_1(\alpha,\beta))&\rightarrow&0, \text{\quad as\quad} N\rightarrow\infty.\label{typeIIerror}
\end{eqnarray*}
\end{theorem}
 
\subsection{Target tracking problems}\label{subsec:egTarget}
In this subsection, we study the target tracking problem. 
We first provide the background towards this problem in Section \ref{subsec:back}.
Then we pose a hypothesis testing problem and apply our theory to it in Section \ref{subsec:targetrevisit}.
\subsubsection{Background}\label{subsec:back}
In this subsection, we discuss another application of the theory.
Let $X_i \in \{0,1\}^{m}$, where $m$ is an integer.
We have $X_{i,j} =0$ or $1$, where $X_{i,j}$ denotes the $j$th
entry of $X_i$. Here $i$ is a time index and $j$ is a location
index. 
$X_{i,j}=1$ (or $0$) corresponds to a target being present
(or absent) at location $j$ and time $i$.

We introduce the following probabilistic model to minic the motion of targets over time. From $X_i$ to $X_{i+1}, 1 \le i \le n-1$, we have:
\begin{enumerate}
\item Initialize $X_{i+1,j} = 0$ for all $j$.

\item If $X_{i,j}=0$, then set $X_{i+1,j}=X_{i+1,j}+1$ with probability
    $p_0$ (corresponding to a newly emerging object).
\item If $X_{i,j}=1$, there are four sub-cases:
\begin{enumerate}
  \item $X_{i+1,j-1} =X_{i+1,j-1} +1$  with probability $p_1$ (shifting left)
  \item $X_{i+1,j  } =X_{i+1,j}+1$  with probability $p_2$ (remain the same
      location)
  \item $X_{i+1,j+1} =X_{i+1,j+1}+1$  with probability $p_3$ (shifting right)
  \item do nothing, with probability $1 - p_1 - p_2 - p_3$ (object
      vanishes).
\end{enumerate}

Apparently, we must impose $0 < p_i < 1$, for $i=0,1,2,3$ and $p_1+p_2+p_3
< 1$.

\item Finally, we take $X_{i+1,j} = \min(1, X_{i+1,j})$ to ensure that each
    one of them is either one or zero. Note $X$ form the ground truth regarding the presence and locations
of the targets.

Below we consider how observations are generated.

\item Set $z_{ij} = x_{ij} + \epsilon_{ij}$, where $\epsilon_{ij}
    \stackrel{\mbox{i.i.d.}}{\sim} N(0,\sigma^2)$, where $\sigma^2$ is a
    parameter, e.g., $\sigma^2=1$.

    Note $Z_i = \{z_{ij}, j=1,2,\ldots, m\}$ is the observation at
    time $i$.

\end{enumerate}

In \cite{AAOMT}, a hidden state Markov process model is mentioned.
In the above case, it is as follows:
\begin{math}
\begin{array}{ccccc ccccc}
X_0 & \to & X_1 & \to & X_2 & \to & X_3 & \to & X_4 & \to \cdots \\
&& \downarrow && \downarrow && \downarrow && \downarrow   \\
& & Z_1 & & Z_2 & & Z_3 & & Z_4 &.
\end{array}
\end{math}
For our purpose, we may not emphasize this Markovian aspect of the problem. Though it is important in the estimation problem.

We pose a hypothesis testing problem in this case, i.e.,
\begin{equation}\label{eq:target}
\mathbb{H}_0: \quad\text{all}\quad X_{i,j}=0\quad\text{versus}\quad\mathbb{H}_1: \quad\text{some}\quad X_{i,j}=1.
\end{equation}
The idea behind the hypothesis testing problem is to say whether there is some newly emerging object at certain location and time or the image just consists of white noisy pixels. 
This null hypothesis setting corresponds to the scenario where there is no target at all time and $p_0$ is set to be $p_0 = 0$.
We will use the theory of the longest chain to solve this problem in the following subsection.

\subsubsection{A revisit using the theory of longest chain}\label{subsec:targetrevisit}
In this part, we will use our theory to estimate an upper bound of the length of the longest significant run in the target tracking problem for an array of size $m$-by-$n$. 
Under the null hypothesis $\mathbb{H}_0$, the image of size $[1,m]\times[1,n]$ is just a white noise image and $Z_{i,j}=\epsilon_{i,j}$ where $ \epsilon_{i,j}\stackrel{\mbox{i.i.d.}}{\sim} N(0,\sigma^2)$. 
For an arbitrary node $Z_{i,j}$ to be significant, we should provide a member threshold $Z^{\ast}$, i.e., the node is significant if $Z_{i,j}>Z^{\ast}$ and insignificant otherwise.
Let $\mathcal{V}$ be the set of nodes under consideration, i.e.,
\[
\mathcal{V}\equiv\{(i,j)\in\mathbb{Z}^2: 1\leq i\leq n, 1\leq j\leq m\}.
\]
Let $\mathcal{E}$ be the set of edges from $(i,j)\in\mathcal{V}$ to $(i+1, j+s)\in\mathcal{V}$ such that $\left|s\right|\leq 1$.
 Let $p=\mathbb{P}(Z_{i,j}>Z^{\ast})$ be the probability of a node $(i,j)\in\mathcal{V}$ to be significant. In order to make a decision, we need to count the length of the longest significant nodes among all the chains along the edges in $\mathcal{E}$, i.e., the chains of the following form
\begin{eqnarray*}
&&\{(i, j_0), (i+1, j_1), \ldots, (i+\ell, j_\ell): \\
&&\left|j_{k+1}-j_{k}\right|\leq 1, k=0,\ldots, \ell-1\}.
\end{eqnarray*}
 
We will use the length of the longest significant run, denoted by $\left|L^{\max}_{T}(m,n)\right|$, as a statistic for the test. 
And a little bit more consideration yields that under the null hypothesis $\mathbb{H}_0$, the chain of significant nodes has the same structure as in (\ref{def:chain}) with $C=1$.

We can apply our theory to find a reasonable threshold. By Theorem \ref{T:p-c}, the critical probability $p_c$ for the graph $(\mathcal{V},\mathcal{E})$ satisfies that $p_c\geq\frac{1}{2C+1}$. 
Therefore, we may choose $Z^{\ast}$ such that $p=\mathbb{P}(Z_{i,j}>Z^{\ast})<\frac{1}{3}$ for $(i,j)\in\mathcal{V}$. Thus if $m$ is constant, by Theorem \ref{th:er-re}, for any $\epsilon_1>0$, there exist $\rho(m,p)\in(0,1)$ and $N\in\mathbb{Z}^{+}$ such that when $n\geq N$, we have
\[
\left|L^{\max}_{T}(m,n)\right|\leq(1+\delta_4)\log_{1/\rho(m,p)}n\text{ with probability } 1-\epsilon_1,
\]
for any $\delta_4>0$. If $m\rightarrow\infty$, $n\rightarrow\infty$, then by Theorem \ref{th:mainconv}, for any $\epsilon_2>0$ we have
\[
\left|L^{\max}_{T}(m,n)\right|\leq(1+\delta_4)\frac{\log(mn)}{\phi(p)}\text{ with probability } 1-\epsilon_2.
\]
So let the decision threshold $\left|L^{\ast}_T\right|$ be $(1+\delta_4)\log_{1/\rho(m,p)}n$ if $n\rightarrow\infty$ with fixed $m$; and $(1+\delta_4)\frac{\log(mn)}{\phi(p)}$ if $n\rightarrow\infty, m\rightarrow\infty$.
Since the forgoing $\epsilon_1$ and $\epsilon_2$ are arbitray, if it happens that
\begin{equation}
\label{eq:targetTest}
\left|L^{\max}_{T}(m,n)\right|>\left|L^{\ast}_T\right|,
\end{equation}
then we can always reject $\mathbb{H}_0$ with false positive probability close to $0$ asymptotically.

In summary, we will have the following corollary.
\begin{corollary}
The longest run test defined in (\ref{eq:targetTest}) can always reject $\mathbb{H}_0$ with type I error close to $0$ asymptotically to the hypothesis testing problem in (\ref{eq:target}).
\end{corollary}

If we have more information on the structure of the observation matrix under the alternative hypothesis (such as condition (\ref{cond:sep})), we can even obtain a stronger result that the longest significant test is asymptotically powerful.

\section{Conclusion and Future work}
\label{sec:conclusion}

In this paper, we  first study  the Pseudo-tree model. 
We find the upper and lower bounds of the asymptotic probability to have a run with length $k$.
By exploring the connection between the  Pseudo-tree model and the inflating Bernoulli net.
We then develop the asymptotic rate of the length of the longest significant run in an inflating Bernoulli net as $m\rightarrow\infty$ and $n\to\infty$. 
 We further apply our theory to the image detection problem to find the reasonable thresholds, which yields a reliable detection procedure. 
It is of interests to learn the value of the function $\phi(p)$ in the future. 
Also for the portion of the nodes in the suspiciously curve, $\epsilon_N$, we develop a lower bound, which guarantees a reliable test. 
However, it remains our future work to find the minimum bound of $\epsilon_N>0$, below which there is no powerful statistical test. 
Also it is not easy to find $\rho(m,p)$ when $m\geq 12$, especially when we drop the independence assumption among the nodes within the same column.

%
%
%
%
%
%

\section{Proofs} 
\label{sec:proofs}

We present all the proofs in this section. 
Our proof techniques in Section \ref{chp:PseudoTreeModel} and Section \ref{Chp:BernoulliNet} mainly come from the percolation theories, association results, and the Chen-Stein Poisson approximations.
The last two techniques are also partly used in \cite{JASA06LongSigRun}.
The proofs in this work are more complicated than in the finite $m$ scenario \cite{JASA06LongSigRun}.

\subsection{Proof of Theorem \ref{T:p-c}}

The proof of Theorem \ref{T:p-c} requires the following definition and lemma. See \cite{GG1989}.
\begin{definition}\label{D:ide}
Let $V$ be the set of nodes in the pseudo-tree and we take the sample space as
\[
\Omega=\prod_{v\in V}\{0,1\}.
\]
We take $\mathcal{F}$ to be the $\sigma$-field of subsets of $\Omega$ generated by the finite-dimensional cylinders. 
We say an event $A\in\mathcal{F}$ is increasing, if the indicator function of $A$ satisfies $I_{A}(X_1)\leq I_{A}(X_2)$ whenever $X_1\leq X_2$, where $X_1, X_2$ are two realizations on $V$, i.e., $X_1:V\rightarrow\{0,1\}$, where $X_1(i,j)=1$ if the node $(i,j)$ is significant and $X_1(i,j)=0$ otherwise and $X_2$ has the same definition. 
Analogously, we say $A$ decreasing set if $\overline{A}$, the complement of $A$, is increasing.
\end{definition}

\begin{lemma}\label{L:FKG} (FKG Inequality)
If $A$ and $B$ are both increasing (or both decreasing) events in the lattice, then we have
$\mathbb{P}(A\bigcap B)\geq\mathbb{P}(A)\mathbb{P}(B)$.
\end{lemma}

The significant edge version of this lemma can be found in Section 2.2 of \cite{GG1989}. 
The intuition behind this lemma is that if there is an open path joining vertex $u$ to vertex $v$, then it becomes more likely that there is an open path joining vertex $x$ to vertex $y$ than without  a path from $u$ to $v$. 
Replacing edge by node in the proof in \cite{GG1989}, the significant node version can be shown analogously.

\begin{proof}[Proof of Theorem \ref{T:p-c}]
Recall that $\theta_k(p)=\mathbb{P}_p(0\leftrightarrow B(k-1))$. The event $\{0\leftrightarrow B(k)\}$ happens if and only if there is an open node $x\in B(1)$, such that the origin $(0,0)$ is open and the event $\{x\leftrightarrow B(k)\}$ occurs. Of course, $\text{card}(B(1))=2C+1$. Therefore we have,
\[
\{0\leftrightarrow B(k)\}=\{\bigcup_{\{x\in B(1)\}}((0\leftrightarrow x)\cap (x\leftrightarrow B(k)))\}
\]
This implies that
\begin{eqnarray}
\label{eq:simp}
\begin{split}
&1-\mathbb{P}_p(0\leftrightarrow B(k))\\
=&\mathbb{P}_p(\bigcap_{\{x\in B(1)\}}\overline{(0\leftrightarrow x)\bigcap
(x\leftrightarrow B(k))})\\
\geq&\prod_{\{x\in B(1)\}}\mathbb{P}_p(\overline{(0\leftrightarrow x)\bigcap (x\leftrightarrow B(k))})\\
=&(1-p\theta_{k}(p))^{(2C+1)},
\end{split}
\end{eqnarray}
where the inequality is due to Lemma \ref{L:FKG} and the fact that $\{0\leftrightarrow x\}\cap\{x\leftrightarrow B(k)\}$
is an increasing event and that $\mathbb{P}_p(x\leftrightarrow B(k))=\mathbb{P}_p(0\leftrightarrow B(k-1))$ for any
given $x\in B(1)$.

So by (\ref{eq:simp}), we have that
\begin{eqnarray*}
\theta_{k+1}(p)&=&\mathbb{P}_p(0\leftrightarrow B(k))\\
&\leq&1-(1-p\mathbb{P}_p(0\leftrightarrow B_{k-1}))^{(2C+1)}\\
&=&1-(1-p\theta_{k}(p))^{(2C+1)}.
\end{eqnarray*}
Given this, we investigate the function
\[
f(x)=1-(1-px)^{r},
\]
where $r\in\mathbb{Z}^{+}$. We have
\[
f'(x)=rp(1-px)^{r-1}>0,\] 
and 
\[ f''(x)=-r(r-1)p^2(1-px)^{r-2}<0, \forall x\in(0,1).
\]
So the function $f(x)$ is always strictly increasing and concave down and $f(0)=0$. 
Besides, one can see that $f'(0)=rp$ and from this we have $f(x)$ is always under the line $y=x$ if $p<\frac{1}{r}$. 
Let $x_0$ be an arbitrary number in $(0,1)$ and generate a sequence $\{x_n\}_{n\geq 0}$, such that $x_{n+1}=f(x_n)$ for $n\geq 0$. 
Since $f(x)<x$ when $x\in(0,1)$ and $p<\frac{1}{r}$, the sequence $\{x_n\}_{n\geq0}$ is strictly decreasing. 
On the other hand, it is easy to see that $x_{n}\geq 0$ for any $n\geq 0$. Because a bounded decreasing sequence must have a limit, we have that
\[
0\leq x^{\ast}\equiv\lim_{n\rightarrow\infty}x_n.
\]
By the continuity of $f(x)$, one can easily see that
\[
f(x^{\ast})=\lim_{n\rightarrow\infty}f(x_n)=\lim_{n\rightarrow\infty}x_{n+1}=x^{\ast}.
\]
Since $f(x)<x$ on $(0,1)$, it is obvious that the limit of the sequence $\{x_n\}_{n\geq 0}$ is 0, i.e., $x^{\ast}=0$ for any starting point $x_0\in(0,1)$. 

Hence when $p<\frac{1}{(2C+1)}$, it leads to
\begin{eqnarray*}
&&\theta(p)\\
&=&\lim_{k\to\infty}\theta_{k+1}(p)\\
&=&\lim_{k\to\infty}\mathbb{P}(0\leftrightarrow B(k))\\
&\leq&\lim_{k\to\infty}1-
(1-p\theta_{k}(p))^{(2C+1)}\\
&=&0.
\end{eqnarray*}
According to the definition $p_c \vcentcolon= \sup\{p\in[0,1]:\theta(p)=0\}$, it follows that $p_c\geq\frac{1}{(2C+1)}$.
\end{proof}

\subsection{Proof of Theorem \ref{T:subexp}}
Before proving the theorem, let us state the sub-additivity lemma which can be found in \cite{GG1989}.
\begin{lemma}\label{L:sub}\textbf{Sub-additive limit theorem.}
If $(x_r: r\geq1)$ is sub-additive, i.e., $x_{m+n}\leq x_m+x_n$ for all $m,n$, then $\lambda=\lim_{r\to\infty}\{\frac{x_r}{r}\}$ exists and satisfies $-\infty\leq\lambda<\infty$. 
Furthermore, we have
\[
\lambda=\inf\{\frac{x_m}{m}:m\geq1\}
\]
and thus $x_m\geq m\lambda$ for all $m$.
\end{lemma}

\begin{proof}[Proof of Theorem \ref{T:subexp}]
Given a positive integer $l$, it is not hard to show that 
\[
\text{card}(B(l))=(2lC+1).
\]
Since the event $\{0\leftrightarrow B(l+k-1)\}$ occurs if and only if there is some $z\in B(l)$ such that both $\{0\leftrightarrow z\}$ and $\{z\leftrightarrow B(l+k-1)\}$ occur, we have
\begin{eqnarray*}
&&\theta_{k+l}(p)=\mathbb{P}_p(0\leftrightarrow B(k+l-1))\\
&=&\mathbb{P}_p(\bigcup_{\{z\in B(l)\}}
(\{0\leftrightarrow z\}\cap\{z\leftrightarrow B(k+l-1)\})) \\
&\leq&\sum_{\{z\in B(l)\}}\mathbb{P}_p(\{0\leftrightarrow z\}\cap\{z\leftrightarrow B(k+l-1)\})\\
&=&\frac{1}{p}\sum_{\{z\in B(l)\}}\mathbb{P}_p(\{0\leftrightarrow z\})\mathbb{P}_p(\{z\leftrightarrow
B(k+l-1)\})\\
&=&\frac{1}{p}\sum_{\{z\in B(l)\}}\mathbb{P}_p(\{0\leftrightarrow z\})\mathbb{P}_p(\{0\leftrightarrow
B(k-1)\}),
\end{eqnarray*}
where the third ``$=$'' is due to conditional probability,
\begin{eqnarray*}
&&\mathbb{P}_p(\{0\leftrightarrow z\}\cap\{z\leftrightarrow B(k+l-1)\})\\
&=&\mathbb{P}_p(\{z\leftrightarrow B(k+l-1)\}|\{0\leftrightarrow z\})\mathbb{P}_p(\{0\leftrightarrow z\})\\
&=&\mathbb{P}_p(\{z\leftrightarrow B(k+l-1)\}|z=1)\mathbb{P}_p(\{0\leftrightarrow z\})\\
&=&\frac{\mathbb{P}_p(\{z\leftrightarrow B(k+l-1)\})}{\mathbb{P}_p(z=1)}\mathbb{P}_p(\{0\leftrightarrow z\})\\
&=&\frac{1}{p}\mathbb{P}_p(\{0\leftrightarrow z\})\mathbb{P}_p(\{z\leftrightarrow B(k+l-1)\})
\end{eqnarray*}
and the last equality is due to the fact that
\[
\mathbb{P}_p(\{z\leftrightarrow B(k+l-1)\})=\mathbb{P}_p(\{0\leftrightarrow B(k-1)\}), \forall z\in B(l).
\]
Notice that $\mathbb{P}_p(0\leftrightarrow z)\leq p\theta_{l}(p)$ for any $z\in B(l)$. We have
\[
\theta_{k+l}(p)\leq(2lC+1)\theta_{l}(p)\theta_{k}(p).
\]
On the other hand, for any $z\in B(l)$, we have
\begin{eqnarray*}
\theta_{k+l}(p)&=&\mathbb{P}_p(\{0\leftrightarrow B(k+l-1)\})\\
&\geq&\mathbb{P}_p(\{0\leftrightarrow z\}\cap \{z\leftrightarrow B(k+l-1)\})\\
&=&\frac{1}{p}\mathbb{P}_p(\{0\leftrightarrow z\})\mathbb{P}_p(\{z\leftrightarrow B(k+l-1)\})\\
&=&\frac{1}{p}\mathbb{P}_p(\{0\leftrightarrow z\})\mathbb{P}_p(\{0\leftrightarrow B(k-1)\}),
\end{eqnarray*}
Notice that
\[
\theta_{l}(p)\leq\frac{1}{p}\mathbb{P}_p(\bigcup_{\{z\in B(l)\}}\{0\leftrightarrow z\})\leq\frac{1}{p}\sum_{\{z\in B(l)\}}
\mathbb{P}_p(\{0\leftrightarrow z\}).
\]
It follows to have a node $z\in B(l)$ such that
\[
\frac{1}{p}\mathbb{P}_p(\{0\leftrightarrow z\})\geq\frac{\theta_{l}(p)}{(2lC+1)}.
\]
Thus we have
\[
\theta_{k+l}(p)\geq\frac{1}{(2lC+1)}\theta_{l}(p)\theta_{k}(p).
\]
If we let $g(l)=\log(2lC+1)$, then the inequalities we get so far are:
\begin{eqnarray*}
\log(\theta_{k+l}(p))&\leq&\log(\theta_{k}(p))+\log(\theta_{k}(p))+g(l); \\
\log(\theta_{k+l}(p))&\geq&\log(\theta_{k}(p))+\log(\theta_{l}(p))-g(l).
\end{eqnarray*}
Notice that $g(k+l)-g(k)=\log(1+\frac{2lC}{2kC+1})\leq\log2$ if $l\leq k$.
Therefore, we have
\begin{eqnarray}
&&\log(\theta_{k+l}(p))+g(k+l)+\log2\nonumber\\
&\leq&\log(\theta_{k}(p))+g(k)+\log2 \label{ineq:subadd}\\
&&+\log(\theta_{l}(p))+g(l)+\log2;\nonumber \\
\nonumber\\
&&\log(\theta_{k+l}(p))-g(k+l)-\log2\nonumber\\
&\geq&\log(\theta_{k}(p))-g(k)-\log2  \label{ineq:supadd}\\
&&+\log(\theta_{l}(p))-g(l)-\log2.\nonumber
\end{eqnarray}
Then by Lemma \ref{L:sub}, we have
\begin{eqnarray*}
\phi(p)&:=&\lim_{k\to\infty}-\frac{1}{k}\{\log(\theta_{k}(p))\}\\
&=&\lim_{k\to\infty}-\frac{1}{k}\{\log(\theta_{k}(p))+g(k)+\log2\}\\
&=&\lim_{k\to\infty}-\frac{1}{k}\{\log(\theta_{k}(p))-g(k)-\log2\}.
\end{eqnarray*}
This leads to
\begin{eqnarray}
\log(\theta_{k}(p))+g(k)+\log2&\geq&-k\phi(p);\label{ineq:upper}\\
-\log(\theta_{k}(p))+g(k)+\log2&\geq&k\phi(p);\label{ineq:lower}
\end{eqnarray}
for all $k\geq 1$. The theorem now follows (\ref{ineq:upper}) and (\ref{ineq:lower}) easily. 
\end{proof}

\subsection{Proof of Corollary \ref{T:thetaratio}}

\begin{proof}
By inequality (\ref{ineq:subadd}), we know that the sequence 
\[
\{\log(\theta_{k}(p))+g(k)+\log2\}_{k\in\mathbb{N}}
\]
is a sub-additive sequence. Thus by Lemma \ref{L:sub} we have
\begin{eqnarray*}
-\phi(p)&=&\lim_{k\to\infty}\frac{\log(\theta_k(p))+g(k)+\log2}{k}\\
&=&\inf_{k\in\mathbb{N}}\frac{\log(\theta_k(p))+g(k)+\log2}{k}
\end{eqnarray*}
Therefore, for any $\epsilon_0>0$, there exists some large $k_0$ such that when $k>k_0$, we have
\[
-\phi(p)\leq\frac{\log(\theta_k(p))}{k}+\epsilon_0,
\]
which leads to
\[
\exp(-\phi(p))\leq(\theta_k(p))^{\frac{1}{k}}\exp(\epsilon_0), \forall k>k_0.
\]
By inequality (\ref{ineq:supadd}), we know that $\{g(k)+\log2-\log(\theta_k(p))\}_{k\in\mathbb{N}}$ is a sub-additive sequence, therefore we have
\begin{eqnarray*}
&&g(k)+\log2-\log(\theta_k(p))\\
&\leq& g(k-1)+\log2-\log(\theta_{k-1}(p))\\
&&+g(1)+\log2-\log(\theta_1(p)).
\end{eqnarray*}
Divide by $k$ on the left and by $k-1$ on the right. It is easy to see that for any $\epsilon_1>0$, there exists some large $k_1$ such that when $k>k_1$, we have
\[
\frac{\log(\theta_k(p))}{k}\geq\frac{\log(\theta_{k-1}(p))}{k-1}-\epsilon_1.
\]
It follows that when $k>\max\{k_0,k_1\}$, we have
\[
\exp(-\phi(p))\leq(\theta_k(p))^{\frac{1}{k}}\exp(\epsilon_0)\leq\frac{\theta_k(p)}{\theta_{k-1}(p)}\exp(\epsilon_0+\epsilon_1/k).
\]
By the same technique using (\ref{ineq:supadd}) and (\ref{ineq:subadd}), we can show that for any $\epsilon_2$, when $k>k_2$ for some large $k_2$, we have
\[
\exp(-\phi(p)+\epsilon_2)\geq(\theta_k(p))^{\frac{1}{k}}\geq\frac{\theta_k(p)}{\theta_{k-1}(p)}.
\]
Since $\epsilon_0$, $\epsilon_1$ and $\epsilon_2$ are arbitray, we have that
\[
\lim_{k\to\infty}\frac{\theta_k(p)}{\theta_{k-1}(p)}=\exp\{-\phi(p)\}.
\]
\end{proof}

\subsection{Proof of Corollary \ref{C:phi}}
Before going to the proof of this corollary, we will first introduce the following lemma.
\begin{lemma}\label{L:nonincre}
Let $A$ be an increasing event which depends only on  finitely many nodes of a lattice. 
Then
$\frac{\log\mathbb{P}_p(A)}{\log p}$ is a non-increasing function of $p$.
\end{lemma}
The proof of this lemma can be found in \cite{GG1989}. 
Though the author in \cite{GG1989} shows the proof in a bond percolation problem, it is very easy to adjust the proof for our purpose and we omit the details.

\begin{proof}[Proof of Corollary \ref{C:phi}]
It is easy to see $\phi(p)=0$ if $p>p_c$. Indeed, since 
\[
\mathbb{P}_{p}(0\leftrightarrow B(k))\geq\theta(p)>0,
\]
it leads to
\begin{eqnarray*}
0&\leq&\phi(p)=\lim_{k\to\infty}-\frac{\log\mathbb{P}_{p}(0\leftrightarrow B(k-1))}{k}\\
&\leq&\lim_{k\to\infty}-\frac{\log\theta(p)}{k}=0,
\end{eqnarray*}
when $p>p_c\geq \frac{1}{(2C+1)}$.

Since $\theta_{k}(p)$ only depends on the status of finitely many sites, $-\frac{1}{k}\log(\theta_{k}(p))$
is a continuous function of $p$ for any $k\geq 1$. So it is sufficient to show that $-\frac{1}{k}\log(\theta_{k}(p))$
converges to $\phi(p)$ uniformly on $(0,1]$. By (\ref{ineq:upper}) and (\ref{ineq:lower}), we have for any $p\in(0,1]$
\[
\left|\phi(p)+\frac{1}{k}\log(\theta_{k}(p))\right|\leq\frac{1}{k}(g(k)+\log2)\rightarrow0\text{\quad as\quad} k\to\infty,
\]
which does not depend on $p$ at all. So $\phi(p)$ is a continuous function of $p$ on $(0,1]$. And it follows the fact that $\phi(p_c)=0$, since
\[
\phi(p_c)=\lim_{p\downarrow p_c}\phi(p)=0,
\]
by continuity of $\phi(p)$. To prove the strict monotonicity of $\phi(p)$ when $0<p<p_c$, we notice that $\{0\leftrightarrow B(k-1)\}$ is an increasing event which only depends on finitely many edges. Thus we apply the Lemma \ref{L:nonincre} to have that
\[
\frac{\log\mathbb{P}_{a}(0\leftrightarrow B(k-1))}{\log a}\geq\frac{\log\mathbb{P}_{b}(0\leftrightarrow B(k-1))}{\log b}
\quad \text{if} \quad a\leq b.
\]
If we divide the above by $k$ and take the limit as $k\rightarrow\infty$, then we have
\[
\phi(a)\geq\phi(b)\frac{\log(\frac{1}{a})}{\log(\frac{1}{b})}\quad\text{if}\quad 0<a\leq b\leq1.
\]
So if $0<a<b<p_c$, we have $\phi(a)>\phi(b)$. 
Thus the function $\phi(p)$ is strictly decreasing on $(0,p_c)$.

To prove that $\lim_{p\to 0}\phi(p)=\infty$, we use $\chi(k)$ and $\chi^{\ast}(k)$ to denote the number of all runs and significant runs respectively in the Pseudo-tree model that connect $0$ and $B(k-1)$ respectively. 
It is not hard to see that $\chi(k)=(2C+1)^{k-1}$ and $\mathbb{E}_{p}(\chi^{\ast}(k))=p^{k}\times\chi(k)$. Therefore, we have the following
\begin{eqnarray*}
\theta_{k}(p)&=&\mathbb{P}_{p}(0\leftrightarrow B(k-1))\\
&=&\mathbb{P}_{p}(\chi^{\ast}(k)\geq 1)\\
&\leq&\mathbb{E}_{p}(\chi^{\ast}(k))\\
&=&p^{k}\times (2C+1)^{k}
\end{eqnarray*}
So this will lead to the following fact
\begin{eqnarray*}
\lim_{k\to\infty}-\frac{\log\theta_{k}(p)}{k}&\geq& -\log(p\times (2C+1)).
\end{eqnarray*}
So as $p\rightarrow0$, obviously $\phi(p)\rightarrow\infty$. 
\end{proof}

\subsection{Proof of Corollary \ref{co:depofrho}}
\begin{proof}
Given a realization 
\[
t_{i,j}\sim\text{Uniform}(0,1), 1\leq i\leq n, 1\leq j\leq m,
\]
let $x^{\ast}_1\geq x^{\ast}_2>0$ be such that $p_1=\mathbb{P}(t_{i,j}>x^{\ast}_1)$ and $p_2=\mathbb{P}(t_{i,j}>x^{\ast}_2)$. Since $t_{i,j}>x^{\ast}_1$ implies that $t_{i,j}>x^{\ast}_2$, one can easily see that each significant node under threshold $x^{\ast}_1$ must be significant under $x^{\ast}_2$ and therefore $\left|L_0(m_1,n,p_1)\right|\leq\left|L_0(m_1,n,p_2)\right|$ which by Theorem \ref{th:er-re} leads to $\rho(m_1,p_1)\leq \rho(m_1,p_2)$. Similarly, it is not hard to see that $\left|L_0(m_1,n,p_1)\right|\leq\left|L_0(m_2,n,p_1)\right|$ since if $m_1\leq m_2$, $([1,n]\times[1,m_1])\cap\mathbb{Z}^2\subset([1,n]\times[1,m_2])\cap\mathbb{Z}^2$. Thus $\rho(m_1,p_1)\leq\rho(m_2,p_1)$.
\end{proof}

\subsection{Proof of Theorem \ref{T:kolmo}}
\begin{proof}
Recall $\theta(p)$, defined in Subsection \ref{defOfTheta}, is the probability that there is a significant run starting from a certain node and heading towards right forever when the probability of a node to be open is $p$. 
Let $C(i,j)$ be a significant run starting at $(i,j)$. In particular, $C$ is the one starting at $(0,0)$. 
The event that there exists an infinite open cluster in the array does not depend on the status of finitely many columns of nodes. 
Thus by the Kolmogorov zero-one law, $\mu(p)$ can only be either $0$ or $1$.
 If $p>p_c$, then of course $\theta(p)>0$. We have 
\[
\mu(p)\geq\mathbb{P}(|C|=\infty)=\theta(p)>0,
\]
which implies that $\mu(p)=1$ by the zero-one law.
On the other hand, because $\mathbb{P}(|C(i,j)|=\infty)=\mathbb{P}(|C|=\infty)=\theta(p), \forall (i,j)$, if $\theta(p)=0$ or $p<p_c$, we have
\[
\mu(p)\leq\sum_{(i,j)}\mathbb{P}(|C(i,j)|=\infty)=0.
\]
\end{proof}

\subsection{Proof of Theorem \ref{th:RhoLimSup}}
\begin{proof}
We first prove
\[
\rho(\infty,p)=\lim_{k\to\infty}\rho_{k}(\infty,p)=\lim_{k\to\infty}\frac{P_{k}(\infty,p)}{P_{k-1}(\infty,p)}=1,
\]
in the case of $p>p_c$. Suppose that $\rho_{k}(\infty,p)$ does not converge to $1$. Then since $[0,1]$
is a compact set, there must exist a subsequence of $\{\rho_{k}(\infty,p)\}_{k\in\mathcal{K}}$, $\mathcal{K}\subset\mathbb{Z}^{+}$, such that
\[
\lim_{k(\in\mathcal{K})\rightarrow\infty}\rho_{k}(\infty, p)=\rho_0,
\]
for some $\rho_0$ in $[0,1)$. And there exists some constant $\rho_{0}^{\ast}\in(0,1)$ slightly bigger than $\rho_0$, such that
\[
\rho_{k}(\infty, p)<\rho_0^{\ast},
\]
for any sufficiently large $k\in\mathcal{K}$ . Therefore, we have
\begin{equation}\label{eq:ptorho}
P_{n}(\infty,p)\leq\prod_{k(\in\mathcal{K})\leq n}\rho_{k}(\infty,p)\leq\prod_{k(\in\mathcal{K})
\leq n}\rho_{0}^{\ast},
\end{equation}
since $P_{n}(\infty,p)$, the probability of having an across when there are exactly $n$ columns, is equal to
$P_{1}(\infty,p)\times\prod_{i=1}^{n}\rho_{i}(\infty,p)$, which is no larger than
\[
\prod_{k(\in\mathcal{K})\leq n}\rho_{0}^{\ast}.
\]
It leads to the fact that
\[
P_{n}(\infty,p)\rightarrow0, \text{\quad as\quad} n\rightarrow\infty.
\]
On the other hand, it is easy to see that $P_{n}(\infty,p)\geq\theta_n(p)\geq\theta(p)>0$ for any $n$ when $p>p_{c}$,
where $\theta_n(p)$ and $\theta(p)$ are defined in subsection \ref{notation}. 
So there is a contradiction. 
Therefore we should have the following,
\[
\rho_{k}(\infty,p)\rightarrow 1, \text{\quad as\quad}k\rightarrow\infty,
\]
when $p>p_c$.

Now we prove the first equality of (\ref{eq:RhoLimSup}) under $p>p_c$. 
By Corollary \ref{co:depofrho}, we know the limit of $\rho(m,p)$ exists as $m$ goes to $\infty$ and thus we have the following
\[
\lim_{m\to\infty}\rho(m,p)=\sup_{m\in\mathbb{Z}^{+}}\{\rho(m,p)\}:=\rho^{\ast}.
\]
If we had $\rho^{\ast}<1$, then notice the fact that $\rho(m,p)\leq\rho^{\ast}$
for every $m$, thus it would lead to
\begin{equation}\label{lim:conv}
\frac{\left|L_{0}(m,n)\right|}{\log_{1/\rho(m,p)}n}\rightarrow 1, \text{\quad as\quad} n\rightarrow\infty,
\end{equation}
almost surely, for every $m$. We would have $\left|L_{0}(m,n)\right|\leq\log_{1/\rho^{\ast}}n$ with probability $1$, when $n$ is sufficiently large for every $m$.
On the other hand, in the array of $\mathbb{Z}^{+}\times\mathbb{Z}$,
we would have positive probability $(\geq\theta_{n}(p)\geq\theta(p)>0)$ that there is a significant run connecting the origin and the $n$th column. 
This leads to the fact that we have an across in the first $n$ columns with positive probability
for any positive integer $n$. 
So given $n$ sufficiently large, we may choose $m(\geq 3n\cdot C)$ to be
sufficiently large such that the model contains all the possible significant runs in the first $n$ columns starting at the origin. 
Therefore with positive probability ($>\theta(p)$), we have an across in the first $n$ columns which contradicts (\ref{lim:conv}) above because $\log_{1/\rho^{\ast}}n\ll n$ when $n$ is large. 
The proof of the theorem is completed.
\end{proof}

\subsection{Proof of Theorem \ref{th:ptoinfty}}
Before the proof, let us recall the definitions of three constants in \cite{twommts}. 
Let $I$ be an arbitrary index set, and for $\alpha\in I$, let $X_{\alpha}$ be a Bernoulli random variable with $p_{\alpha}\equiv\mathbb{P}(X_{\alpha}=1)=1-\mathbb{P}(X_{\alpha}=0)>0$. 
For each $\alpha\in I$, suppose we have chose $B_{\alpha}\subset I$ with $\alpha\in B_{\alpha}$. We think of $B_{\alpha}$ as a \lq\lq{}neighborhood of dependence\rq\rq{} for $\alpha$, such that $X_{\alpha}$ is independent or nearly independent of all of the $X_{\beta}$ for $\beta$ not in $B_{\alpha}$. 
Define
\begin{eqnarray*}
b_1&\equiv&\sum_{\alpha\in I}\sum_{\beta\in B_{\alpha}}p_{\alpha}p_{\beta},\\
b_2&\equiv&\sum_{\alpha\in I}\sum_{\alpha\not=\beta\in B_{\alpha}}p_{\alpha\beta}, \text{ where } p_{\alpha\beta}=\mathbb{E}(X_{\alpha}X_{\beta}),\\
b_3&\equiv&\sum_{\alpha\in I}s_{\alpha},
\end{eqnarray*}
where 
\[
s_{\alpha}\equiv\mathbb{E}\left|\mathbb{E}\{X_{\alpha}-p_{\alpha}\big|\sigma(X_{\beta}: \beta\in I-B_{\alpha})\}\right|.
\]
The following theorem can be found  in \cite{twommts}.

\begin{theorem}
When $b_1 + b_2 +b_3 \rightarrow 0$, the random variable defined by
\[
W\equiv\sum_{\alpha\in I}X_{\alpha},
\]
approximately has a Poisson distribution with mean
\[
\lambda\equiv\mathbb{E}W=\sum_{\alpha\in I}p_{\alpha}.
\]
\end{theorem}
\begin{proof}[Proof of Theorem \ref{th:ptoinfty}]
Let $Z_i$ be the indicator that there is a significant run from $(i,1)$ to the $n$th column, where $i=1,\ldots,m$. 
Let $W_{n,m}$ be the number of nodes in the first column from which an across significant run starts, i.e.,
\[
W_{n,m}=\sum_{i=1}^{m}Z_i.
\]
Obviously that
\[
P_{n}(m,p)=1-\mathbb{P}(W_{n,m}=0).
\]
The main idea of the Poisson approximation is that under certain conditions 
\[
\mathbb{P}(W_{n,m}=0)
\]
can be approximated by Poisson$(\lambda)$ where the Poisson parameter $\lambda$ will be computed below.

To verify the conditions for the Poisson approximation, we first define the neighborhood of $i$, $1\leq i\leq m$, as
\[
N(i)=\{j:|i-j|<2\cdot n\cdot C+1, 1\leq j\leq m\}.
\]
Define three constants $b_1$, $b_2$ and $b_3$ as in \cite{twommts} which depend on $n$, $m$, $C$ and {p}. 
Let $\sigma(Z_j: Z_j\not\in N(i))$ be the $\sigma$-algebra generated by $\{Z_j: Z_j\not\in N(i)\}$. If $j\not\in N(i)$, then clearly $|j-i|\geq 2\cdot n\cdot C+1$ which leads to the fact that $Z_{i}$ and $Z_{j}$ are independent. For $b_3$, we have
\begin{eqnarray*}
b_3&=&\sum_{i=1}^{m}\mathbb{E}\left|\mathbb{E}(Z_i-\mathbb{E}(Z_i))\big|\sigma(Z_j: Z_j\not\in N(i))\right|\\
&=&0,
\end{eqnarray*}

For $b_1$, we have
\begin{eqnarray*}
b_1 &=& \sum_{i=1}^{m}\sum_{j\in N(i)}\mathbb{E}_{p}(Z_i)\mathbb{E}_{p}(Z_j)\\
       &=& \sum_{i=1}^{m}\sum_{j\in N(i)}\mathbb{P}_{p}(Z_{i}=1)\mathbb{P}_{p}(Z_{j}=1)\\
       &\displaystyle \stackrel{\text{(fact 1)}}\leq& \sum_{i=1}^{m}\theta_{n}(p)\sum_{j\in N(i)}\theta_{n}(p)
\end{eqnarray*}
By Theorem \ref{T:subexp}, when $p<p_c$, we have a constant $\sigma>0$ and $\phi(p)>0$ such that 
\[
\theta_n(p)\leq \sigma\cdot n\exp\{-n\phi(p)\}.
\]
And therefore, it follows that
\begin{eqnarray*}    
b_1 &\leq& m\cdot (2n\cdot C+1)\cdot\sigma^{2}\cdot n^{2}\cdot\exp\{-2n\cdot\phi(p)\}\\
       &\leq& O(n^3\cdot\exp\{-n\cdot(\delta_2+\phi(p))\}).
\end{eqnarray*}

For $b_2$, we have
\begin{eqnarray*}
b_2 &=&\sum_{i=1}^{m}\sum_{j\in N(i),j\not=i}\mathbb{E}_{p}(Z_{i}\cdot Z_{j})\\
       &=&2\sum_{i=1}^{m}\sum_{j\in N(i),j>i}\mathbb{E}_{p}(Z_{i}\cdot Z_{j})\\
       &=&2\sum_{i=1}^{m}\sum_{j\in N(i),j>i}\mathbb{P}_{p}(Z_i=1 \text{\quad and\quad}Z_j=1)\\
       &=&2\sum_{i=1}^{m}\mathbb{P}_{p}(Z_i=1)\cdot\sum_{j\in N(i),j>i}\mathbb{P}_{p}(Z_j=1|Z_i=1)\\
      &\leq& 2\sum_{i=1}^{m}\mathbb{P}_{p}(Z_i=1)\cdot (n\cdot C+1)\\
      &\displaystyle \stackrel{\text{(fact 1, 2)}}\leq& O(m\cdot n^{2}\cdot\exp\{-n\cdot\phi(p)\})\\
      &\displaystyle \stackrel{\text{(fact 3)}}\leq& O(n^2\cdot\exp\{-n\cdot\delta_2\}).
\end{eqnarray*}

In the foregoing, we have used the following fact:
\begin{enumerate}
\item $\mathbb{P}(Z_i=1)\leq\theta_{n}(p), \forall i=1,\ldots, m$ and $0<p<p_c$;
\item $\theta_{n}(p)\leq O(n\cdot\exp\{-n\cdot\phi(p)\})$;
\item $O(n^{1+\delta_1})\leq m\leq O(\exp\{n\cdot(\phi(p)-\delta_2)\})$ for some sufficiently small $\delta_1, \delta_2>0$.
\end{enumerate}
Loosely speaking, $b_1$ measures the neighborhood size, $b_2$ measures the expected number of neighbors of a given occurrence and $b_3$ measures the dependence between an event and the number of occurrences outside its neighborhood. 
Now let us consider the Poisson parameter $\lambda$ which is $\mathbb{E}(W_{n,m})$. 
When $O(n^{1+\delta_1})\leq m$ for some sufficiently small $\delta_1>0$, by Theorem \ref{T:subexp} it is easy to see that
\[
\lambda \approx m\theta_n(p)=m\cdot\exp\{-n\cdot(\phi(p)+o(1))\}
\]
as $m$ sufficiently large since $O(n^{1+\delta_1})\leq m$ is enough to relieve the boundary effects. By Theorem 1 of \cite{twommts}, the Poisson approximation gives
\begin{eqnarray*}
&&|\mathbb{P}(W_{n,m}=0)-\exp\{-\lambda\}|\\
&\leq&\min\{1,\frac{1}{\lambda}\}\cdot(b_1+b_2+b_3)\\
 &\leq& O(n^2\cdot\exp\{-n\cdot(\delta_2+\phi(p))\})+O(n^2\cdot\exp\{-n\cdot\delta_2\})\\
 &\leq& O(n^2\cdot\exp\{-n\cdot\delta_2\}).
\end{eqnarray*}
Therefore, under the sub-critical phase, i.e., $p<p_c$, if $m,n$ are sufficiently large with $O(n^{1+\delta_1})\leq m\leq O(\exp\{n\cdot(\phi(p)-\delta_2)\})$, then we have
\begin{eqnarray}
&&\mathbb{P}(W_{n,m}=0)\\
&=&\exp\{-m\theta_n(p)\}+o(1)\\
&=&\exp\{-m\cdot\exp\{-n\cdot(\phi(p)+o(1))\}\}+o(1).
\end{eqnarray}

Note that $m\theta_n(p)=m\exp\{-n(\phi(p)+o(1))\}\leq O(\exp\{-n(\delta_2+o(1)\})$ can be sufficiently small if $n$ is sufficiently large.
Since $1-\exp\{-x\}= x+O(x^2)$ as $x\rightarrow 0$, when $p<p_c$ by Corollary \ref{T:thetaratio} we have
\begin{eqnarray*}
\rho_n(m,p)&=&\frac{P_{n}(m,p)}{P_{n-1}(m,p)}\\
    &=&\frac{1-\exp\{-m\theta_n(p)\}+o(1)}{1-\exp\{-m\theta_{n-1}(p)\}+o(1)}\\
    &=&\frac{m\theta_n(p)+O(m^2\theta^2_n(p)\})+o(1)}{m\theta_{n-1}(p)+O(m^2\theta^2_{n-1}(p)\})+o(1)}\\
    &\rightarrow&\frac{\theta_n(p)}{\theta_{n-1}(p)}\rightarrow\exp\{-\phi(p)\}.
\end{eqnarray*}
as $m\rightarrow\infty, n\rightarrow\infty$ and $(m,n)\in\mathcal{A}_{c_1,c_2,\delta_1,\delta_2}$.
\end{proof}

\subsection{Proof of Theorem \ref{th:mainconv}}
\begin{proof}
This proof was first used in \cite{ClustDetecPerco} for a regular lattice model. Let $\theta^x_k(p)$ denote the probability that $x=(x_1,x_2)\in([1,m]\times[1,n])\cap\mathbb{Z}^2$ connects the $x_2+k-1$-th column, which is denoted by $B(x_2+k-1)$, with a significant chain. One can easily see that $\theta^x_k(p)=\theta_k(p)$. Recall the definition of $\phi(p)$ in the following,
\[
\phi(p)=-\lim_{k\to\infty}\frac{\log\theta_k(p)}{k}=-\lim_{k\to\infty}\frac{\log\theta^x_k(p)}{k}.
\]
Let $\epsilon<1/2$ be a small positive number and $k^+_{m,n}(\epsilon)=\lceil(1+\epsilon)\log(mn)/\phi(p)\rceil$. By the second inequality in (\ref{eq:subexp}), it is not hard to see that
\begin{eqnarray*}
&&\mathbb{P}(\left|L_0(m,n)\right|>k^+_{m,n}(\epsilon))\\
&=&\mathbb{P}(\bigcup_{x\in([1,m]\times[1,n])\cap\mathbb{Z}^2}(x\leftrightarrow B(x_2+k^+_{m,n}(\epsilon)-1)))\\
&\leq&\sum_{x\in([1,m]\times[1,n])\cap\mathbb{Z}^2}\mathbb{P}(x\leftrightarrow B(x_2+k^+_{m,n}(\epsilon)-1))\\
&\leq&mn\sigma_2 k^+_{m,n}(\epsilon)\exp\{-k^+_{m,n}(\epsilon)\phi(p)\}
\end{eqnarray*}
Since $\sigma_2$ is a constant and $\phi(p)>0$ when $p<p_c$, when $m$ and $n$ are sufficiently large, it follows that
\begin{eqnarray*}
&&mn\sigma_2 k^+_{m,n}(\epsilon)\exp\{-k^+_{m,n}(\epsilon)\phi(p)\}\\
&\leq&mn\exp\{-(1-\epsilon/2)k^+_{m,n}(\epsilon)\phi(p)\}\\
&\leq&mn\exp\{-(1-\epsilon/2)(1+\epsilon)\log(mn)\}\\
&=&(mn)^{-(\epsilon-\epsilon^2)/2}\\
&\rightarrow&0, \text{\quad as\quad} m, n\rightarrow\infty
\end{eqnarray*}
since $\epsilon-\epsilon^2>0$.

On the other hand, let $I=\lceil\frac{mn}{(\log m\log n)^2}\rceil$ and let $k^-_{m,n}(\epsilon)$ be $\lfloor(1-\epsilon)\log (mn)/\phi(p)\rfloor$. Let $x^1, x^2\ldots, x^I\in([1,m]\times[1,n])\cap\mathbb{Z}^2$ be nodes separated from each other and the boundary of $([1,m]\times[1,n])\cap\mathbb{Z}^2$ by at least $\frac{1}{2}(\log m\log n)^2$. For sufficiently large $m$ and $n$, it is not hard to see that the $I$ events $\{x^i\leftrightarrow B(x^i_2+k^-_{m,n}(\epsilon)-1)\}$ are independent and have equal probabilities. Therefore, for large $m$ and $n$, by the first inequality of (\ref{eq:subexp}) we have that
\begin{eqnarray*}
&&\mathbb{P}(\left|L_0(m,n)\right|<k_{m,n}(\epsilon))\\
&\leq&\mathbb{P}\left(\bigcap_{i=1,\ldots, I}\{\overline{x^i\leftrightarrow B(x^i_2+k^-_{m,n}(\epsilon)-1)}\}\right)\\
&=&\left(1-\mathbb{P}(x^i\leftrightarrow B(x^i_2+k^-_{m,n}(\epsilon)-1))\right)^I\\
&=&(1-\theta_{k^-_{m,n}(\epsilon)}(p))^I\\
&\leq&(1-\sigma_1 (k^{-}_{m,n}(\epsilon))^{-1}\exp\{-k^-_{m,n}(\epsilon)\phi(p)\})^I
\end{eqnarray*}
When $m$ and $n$ are sufficiently large, it follows that
\begin{eqnarray*}
&&(1-\sigma_1 (k^{-}_{m,n}(\epsilon))^{-1}\exp\{-k^-_{m,n}(\epsilon)\phi(p)\})^I\\
&\leq&(1-\exp\{-(1+\epsilon/2)k^-_{m,n}(\epsilon)\phi(p)\})^I\\
&\leq&(1-\exp\{-(1+\epsilon/2)(1-\epsilon)\log mn\})^I\\
&\leq&(1-(mn)^{-1+\epsilon/2+\epsilon^2/2})^{mn/(\log m\log n)^2}\\
&\leq&\left(1 - (mn)^{-1+\epsilon/2}\right)^{(mn)^{1-\epsilon/2}(mn)^{\epsilon/2}/(\log m\log n)^2}\\
&\leq&\exp\{-(mn)^{\epsilon/2}/(\log m\log n)^2\}\\
&\rightarrow&0, \text{\quad as \quad} m,n\rightarrow\infty.
\end{eqnarray*}
Therefore, as $m, n\rightarrow\infty$, we have $\frac{\left|L_0(m,n)\right|}{\log mn}\rightarrow\frac{1}{\phi(p)}$ in probability.
\end{proof}

\bibliographystyle{IEEEtran}
\bibliography{lrref}

%

%
%
%




\end{document}